\documentclass{amsart}
\usepackage{amsfonts}
\usepackage{amsmath}
\usepackage{amssymb}
\usepackage{amsthm}
\usepackage{color}
\usepackage{comment}

\usepackage{pgfplots}

\theoremstyle{theorem}
\newtheorem{theorem}{Theorem}
\newtheorem{proposition}[theorem]{Proposition}
\newtheorem{lemma}[theorem]{Lemma}
\newtheorem{corollary}[theorem]{Corollary}

\theoremstyle{definition}
\newtheorem{definition}[theorem]{Definition}
\newtheorem{remark}[theorem]{Remark}

\newtheorem{example}[theorem]{Example}

\def\r{{\textbf{r}}}
\def\p{{\textbf{p}}}
\def\b{{\textbf{b}}}

\def\a{\alpha}
\def\bb{\beta}
\def\l{\lambda}

\def\cT{\mathcal T}
\def\cC{\mathcal C}
\def\cR{\mathcal R}
\def\cS{\mathcal S}

\def\NN{{\mathbb{N}}}

\def\ovl#1{\overline{#1}}
\def\t#1{\widetilde{#1}}
\def\disp{\displaystyle }

\author[G.~Ch\`eze]{Guillaume Ch\`eze}
\author[E.~Fieux]{Etienne Fieux}
\address{Guillaume Ch\`eze, Etienne Fieux: Institut de Math\'ematiques de Toulouse\\
Universit\'e Paul Sabatier \\
118 route de Narbonne\\
31 062 TOULOUSE cedex 9, France}
\email{guillaume.cheze@math.univ-toulouse.fr}
\email{etienne.fieux@math.univ-toulosue.fr}

\begin{document}

\title{The Inversion Paradox and Ranking Methods in Tournaments}
\markright{The Inversion Paradox and Ranking Methods}

\maketitle

\begin{abstract}
This article deals with ranking methods. We study the situation where a tournament between $n$ players $P_1$, $P_2$, \ldots $P_n$ gives the ranking  $P_1 \succ P_2 \succ \cdots \succ P_n$, but, if the results of $P_n$ are no longer taken into account (for example $P_n$ is suspended for doping), then the ranking becomes $P_{n-1} \succ P_{n-2} \succ \cdots \succ P_2 \succ P_1$. If such a situation arises, we call it an inversion paradox. In this article, we give a sufficient condition for the inversion paradox to occur. More precisely, we give an impossibility theorem. We prove that if a ranking method satisfies three reasonable properties (the ranking must be natural, reducible by Condorcet tournaments and satisfies the long tournament property) then we cannot avoid the inversion paradox, i.e., there are tournaments where the inversion paradox occurs. We then show that this paradox can occur when we use classical methods, e.g., Borda, Massey, Colley and Markov methods.
\end{abstract}


\section{Introduction.}

Ranking items following the results of a tournament is a classic and ancient problem. It has appeared in economics, politics as well as in sport, \cite{who1,Laslier}.
This problem has a long history and numerous ranking methods have been proposed, for example the Borda method, the Massey method, the Colley method, the Markov method, the Elo ranking, etc. For a given situation, i.e., at the end of a tournament, all these methods do not necessarily give the same ranking.

  Here, we will focus on a particular problem: the inversion paradox, or more precisely the paradox of inversion after elimination of the last ranked player.
To explain this paradox, let us consider a situation where different players compete in a sports tournament. Let us imagine that after a tournament we use a ranking method and obtain the following result: $P_1 \succ P_2 \succ \cdots\succ P_n$. This means that player $P_1$ is ranked first, $P_2$ second and $P_n$ last. Now let us imagine that player $P_n$ has been suspended for doping. In this case, we have to delete the results of the matches in which the player $P_n$ met another player and carry out a new ranking. On the face of it, since $P_n$ was ranked last, it should not  matter. However, for certain methods such as those of Borda, Massey, Colley, and Markov, it is possible that after deleting $P_n$ the new ranking becomes $P_{n-1} \succ P_{n-2} \succ \cdots \succ P_2 \succ P_1$. This ranking gives the reverse order between the remaining players! This is what we call the inversion paradox.


When we consider the Markov algorithm that gave rise to Google's PageRank method, the inversion paradox can be interpreted as follows: the players are the web pages and the matches between the players are replaced by the hyperlinks between the pages. After searching the web, we get a ranking of pages related to our search. Around ten pages are displayed on the screen. Some pages have been used for the ranking but are not displayed. Intuitively, we think that removing the last page in the ranking from the database does not change the ranking. However, the inversion paradox can appear when we use the Markov method. This means that if we delete the last page, then the order in which the results of our search are displayed can be reversed. In other words, deleting the last page (which is not even displayed on the screen) can completely change the way the result is displayed.
For example, let us imagine that a query produces a list of 100 classified pages, which we will write down as $P_1$, $P_2$, \ldots, $P_{100}$. Let us also suppose that 10 pages are displayed on the screen. The screen then shows pages $P_1$, $P_2$, \ldots, $P_{10}$.  If the inversion paradox occurs and the 100th page is deleted, we get on the screen pages $P_{99}$, $P_{98}$, \ldots , $P_{90}$, in this order. The user then gets a completely different result. He or she cannot even see that the order has been reversed, because he or she cannot see all the pages on the screen.\\

The inversion paradox was studied in the voting situation by Peter C. Fishburn 
in \cite{Fishburn} and then by Donald G. Saari, see for example \cite{Saari}. Here, in Theorem \ref{thm:inv}, we will give a sufficient condition 
for a ranking method  to give rise to an inversion
paradox. This theorem can be stated as an impossibility theorem. Indeed, our theorem says that a \emph{natural} ranking method cannot satisfy \emph{the long tournament property}, \emph{reducibility by Condorcet tournaments},  and avoid the inversion paradox, i.e., there are tournaments where the inversion paradox occurs. (Emphasized terms will be defined in the following).

 Impossibility results for tournaments have already been proved, see for example \cite{Altman_Tennenholtze,Csato_impossibility_paired,Csato_impossibilities_generalized_tournaments}. Furthermore, the impact of eliminating an extremely strong or extremely  weak candidate during a championship has already been studied by Aleksei Y. Kondratev et al. in \cite{Kondratev}. They consider generalized scoring rules and they illustrate their problem with the case of the 2014/2015 Biathlon World Cup: the removal of an athlete for doping in 2019 changed the final ranking four years after the competition!
Here, we are considering a more general situation. In our situation, each athlete does not necessarily take part in exactly the same number of races, competitions or matches. In addition, we study the situation in which we remove the last player from the final ranking. Kondratev et al. work on the situation where the eliminated player has lost all his or her  matches.

 As a corollary of Theorem~\ref{thm:inv}, we  obtained  that there are  tournaments where the methods of Borda, Massey, and 
Colley  give an inversion paradox.
We will also give examples for which the Markov method also causes 
such a paradox. All these examples are constructed on the same tournament which comes from the  proof of the theorem.


\subsection{Notations and definitions.}
In this article, we study tournaments and ranking methods. In the following, we give  formal definitions of these notions.

\begin{definition}
Consider a set of $n$ players, $P_1$, $P_2$, \ldots, $P_n$.\\
A match is a quadruple $(P_i,P_j,s_i,s_j)$, where $P_i$ and $P_j$ are the names of the two players and $s_i$, $s_j$ are the score obtained by $P_i$ and $P_j$ during this match.\\
A tournament is a finite list of matches.
\end{definition}

As done in \cite{chartier_etall}, we define a perfect tournament as a tournament where every player plays every other player once and there are no upsets. For example, if we have $n$ players $P_1$, \ldots, $P_n$ then an example of a perfect tournament is given by $P_i$ beats $P_j$ if and only if $i<j$. In other words:
 $P_1$ beats all the other players, $P_2$ is defeated by $P_1$ but beats all the other players. $P_3$ is defeated by $P_1$ and $P_2$ but beats all the other players, and so on. Furthermore, we suppose that each match ends with the score $1-0$.
More generally, we have the following definition.

\begin{definition}
We say that a tournament is a perfect tournament if each player meets all the other players just once and  there exists $\sigma$, a permutation of $\{1, \ldots, n\}$, such that
 each match has the following form
$$(P_{\sigma(i)},P_{\sigma(j)},1,0), \textrm{ where } i<j.$$
 We denote this perfect tournament by $\mathcal{T}_{P_{\sigma(1)} \succ P_{\sigma(2)} \succ \cdots \succ P_{\sigma(n)}}$.\\
\end{definition} 
 
In the perfect tournament $\mathcal{T}_{P_{\sigma(1)} \succ P_{\sigma(2)} \succ \cdots \succ P_{\sigma(n)}}$, the player $P_{\sigma(1)}$ beats all the others players, $P_{\sigma(2)}$ is defeated by $P_{\sigma(1)}$ but beats $P_{\sigma(i)}$ for all $i>2$, etc.\\

If we have two tournaments $\mathcal{T}_1$ and $\mathcal{T}_2$ then the union of these tournaments consists of all the results obtained in $\mathcal{T}_1$ and $\mathcal{T}_2$. We denote this tournament by $\mathcal{T}_1 + \mathcal{T}_2$.
This notation is useful when we consider the union of $k$ identical tournaments. More precisely, if we have a tournament $\mathcal{T}_1$, a tournament $\mathcal{T}_2$, and the results in these two tournaments are exactly the same, then we can write $\mathcal{T}_1=\mathcal{T}_2=\mathcal{T}$. Then, the union of these tournaments is denoted by $\mathcal{T}_1+\mathcal{T}_2=\mathcal{T}+\mathcal{T}=2\mathcal{T}$.\\

Some tournaments have a special structure. For example, in the tournament
$$\mathcal{T}_{P_1\succ P_2 \succ P_3} + \mathcal{T}_{P_2 \succ P_3 \succ P_1} + \mathcal{T}_{P_3\succ P_1 \succ P_2}$$
\begin{itemize}
\item[-]$P_1$ wins two matches against $P_2$ and loses two matches against $P_3$,
\item[-]$P_2$ wins two matches against $P_3$ and loses two matches against $P_1$,
\item[-]$P_3$ wins two matches against $P_1$ and loses two matches against $P_2$.
\end{itemize}
This previous tournament is an example of what we call a Condorcet tournament.

\begin{definition}
Given $n$ players and a permutation $\sigma$, a Condorcet tournament is a tournament with the following form
$$\mathcal{C}_{P_{\sigma(1)}, P_{\sigma(2)}, \ldots, P_{\sigma(n)}}
:= \sum_{i=0}^{n-1} \mathcal{T}_{P_{\sigma(1+i [n])}\succ P_{\sigma(2+i [n])} \succ \cdots \succ P_{\sigma(n+i [n])}}$$
where $k+i[n]$ is the residue modulo $n$ of $k+i$ in the interval $[1,n]$.
\end{definition}

For example, if $n=4$, $\sigma(1)=2$, $\sigma(2)=1$, $\sigma(3)=4$, and $\sigma(4)=3$, then 
\begin{eqnarray*}
\mathcal{C}_{P_{\sigma(1)},P_{\sigma(2)},P_{\sigma(3)}, P_{\sigma(4)}}&=&\mathcal{C}_{P_2,P_1,P_4,P_3}\\
&=& \mathcal{T}_{P_2\succ P_1 \succ P_4 \succ P_3} + \mathcal{T}_{P_1 \succ P_4 \succ P_3 \succ P_2} + \mathcal{T}_{P_4\succ P_3 \succ P_2 \succ P_1}\\
&& + \mathcal{T}_{P_3\succ P_2 \succ P_1 \succ P_4}.
\end{eqnarray*}

In this article we study different ranking methods. 

\begin{definition}
A ranking is  a weak ordering on  the set of players.\\
A ranking method is a function $f$ which takes as input 
a tournament $\mathcal{T}$ and outputs $f(\mathcal{T})$ 
the ranking between the different players.
\end{definition}
In the following, if player $P_2$ is ranked first, $P_3$ second, and $P_1$ third, we will denote this situation by $P_2 \succ P_3 \succ P_1$. There exist different ranking methods, see e.g., \cite{who1,Stefani}. In this article, we study in particular the Borda, Massey, Colley, and Markov methods.


\subsection{Some classical methods.}
In this section we recall briefly how the Borda, Massey, 
Colley and Markov methods work. The idea behind these methods is to associate to each 
player a rating and then deduce a ranking from this 
rating. Of course, the higher the rating, the better the rank.

\subsection{The Borda method.}
A simple way of associating a ranking with a player is to count the number of victories. The player with the most wins is ranked first, the player with the second most wins is ranked second, etc.

In this article, we call this method the Borda method. Indeed, in 1781, Jean-Charles de Borda has suggested a voting system that is referred to as the Borda count, see~\cite{Borda}. The Borda method works as follows. Each voter writes on a ballot the order of his or her preferences for the candidates, e.g., $P_1\succ P_2 \succ \cdots \succ P_n$ means that $P_1$ is preferred to $P_2$, which in turn is preferred to $P_3$, and so on. A candidate gains as many points as the number of candidates behind him. So, in our example, $P_1$ wins $n-1$ points, $P_2$ wins $n-2$ points, ..., $P_{n-1}$ wins 1 point, and $P_n$ wins 0 point. Each ballot paper therefore awards points to each candidate. The points obtained by each candidate are then added together, and the candidate with the most points wins the election.

In the case of a perfect tournament $P_1 \succ P_2 \succ \cdots \succ P_n$, $P_1$ has $n-1$ victories, $P_2$ has $n-2$ victories, \ldots, $P_1$ has 1 victory and $P_n$ has 0 victories. Thus, counting the number of victories coincides with Borda's method. 

\begin{example}
We  consider a tournament where the result of each match is $1-0$ or $0-1$.  The results of this tournament are summarised in the following matrix $B$ whose entry $B_{ij}$ gives the number of wins of $P_i$ against $P_j$
$$B=\begin{pmatrix}
0&1&1&5\\
1&0&1&2\\
2&2&0&1\\
2&2&2&0
\end{pmatrix}.
$$
For example, this matrix means that $P_1$ wins five times and loses twice against $P_4$.
In this situation, $P_1$ has 7 victories, $P_2$ has 4 victories, $P_3$ has 5 victories, and $P_4$ has 6 victories. Thus the Borda method gives the ranking: $P_1 \succ P_4 \succ P_3 \succ P_2$.
\end{example}

\subsection{The Massey method.}
In 1997, Kenneth Massey proposed a method for 
ranking college football teams \cite{Massey}. This method  has been a part of the Bowl Championship Series since the 1999 season. 
Massey's method is based on the following idea. 
We would like to have for each match between $P_i$ and $P_j$
$$r_i-r_j=s_i-s_j,$$
where $r_i$ and $r_j$ are the ratings  of the two players and $s_i$ and $s_j$ are 
their scores.
All the matches in the tournament give a linear system of equations $X\r= \textbf{s}$, where the $i$-th component of $\r$ is $r_i$ 
and $X$ is a matrix whose coefficients belong to $\{-1;0;1\}$.
As this system is overdetermined, it does not necessarily have a solution. Thus we compute the least squares solution of this system. 
This means that we consider what is called the Massey matrix $M=X^TX$ 
and we solve $$M\r=\p$$ where $\p=X^T\textbf{s}$ is the cumulative point differential vector. Thus, the entries $M_{ij}$ of $M$  are given by
$$M_{ij}= \begin{cases}
\textrm{ $-$(number of matches played by } P_i \textrm{ against } P_j), \textrm{ if } i \neq j,\\
\textrm{ total number of matches played by  } P_i, \quad \quad  \quad  ~~~ \textrm{if } i=j.
\end{cases}$$

\noindent The entries $p_i$ of the cumulative differential vector $\p$ are defined by
$$p_i= f_i-a_i,$$
where $f_i$ is the total points scored by $P_i$  during the tournament and $a_i$ is the total points scored against $P_i$  during the tournament. 

\noindent The Massey matrix associated to a perfect tournament with $n$ players is the following  $n \times n$ matrix:
$$\begin{pmatrix}
n-1 & -1 & -1 & \ldots & -1\\
-1 & n-1 &-1 & \ldots & -1\\
-1&-1 &n-1 &\ldots &-1\\
\vdots & \vdots & & \ddots &\\
-1 &-1 &-1 & \ldots & n-1
\end{pmatrix}.
$$


\noindent The cumulative point differential vector $\p$ associated to the perfect tournament  $\mathcal{T}_{P_1 \succ \cdots \succ P_n}$ is given by $\p=\big(n-1, n-3, n-5, \ldots, -(n-3),-(n-1)\big)^T$.\\


The Massey matrix $M$ is singular. Indeed, the sum of the rows is zero.
Therefore, the system $M\r=\p$ has an infinite number of solutions. 
All these solutions give the same ranking between the players. 
We then consider the unique solution $\r$ such that $\sum_{i=1}^nr_i=0$. In order to get this unique solution $\r$, we consider the adjusted Massey matrix $\overline{M}$ obtained by replacing the last row of $M$ with a row of all $1's$. This matrix is invertible and we get the desired rating vector $\r$ by solving the linear system
$$\overline{M}\r=\overline{\p}$$
where $\overline{\p}$ is the vector obtained by replacing the last element of $\p$ with a $0$.


\begin{example}
We consider the same tournament as the one used to illustrate the Borda method. In this situation, the adjusted Massey matrix $\overline{M}$ and the vectors $\overline{\p}$ and $\r$  are:
$$\overline{M}=
\begin{pmatrix}
12&-2&-3&-7\\
-2&9&-3&-4\\
-3&-3&9&-3\\
1&1&1&1
\end{pmatrix},
\quad
\overline{\p}=\begin{pmatrix}
2\\-1\\1\\0
\end{pmatrix}, \quad 
\r=\begin{pmatrix}
0.1142\ldots\\ -0.1009\ldots \\0.0833\ldots \\-0.0966 \ldots\\
\end{pmatrix}.\\
$$
Thus the Massey method gives the ranking $P_1 \succ P_3 \succ P_4 \succ P_2$. 
\end{example}

\subsection{The Colley method.}
In 2001, Wesley Colley wrote a paper about his new method for ranking sports teams \cite{Colley}. This is one of the methods recognized by the NCAA in its list of national champion selectors in college football. The main idea of this method is to compute the winning percentage of each player thanks to Laplace's rule of succession. Colley's method can also be succinctly summarized with a linear system.\\
Here, the rating vector $\r$ is obtained as the solution of the linear system
$$C\r=\b.$$
The matrix $C$ is equal to $M+2I$, where $M$ is the Massey matrix and $I$ is the identity. Thus, the entries $C_{ij}$ of $C$ are given by
$$C_{ij}= \begin{cases}
\textrm{ $-$(number of matches played by } P_i \textrm{ against } P_j), \textrm{ if } i \neq j,\\
2+\textrm{ total number of matches played by  } P_i, \quad \quad \textrm{if } i=j.
\end{cases}$$
The matrix $C$ is invertible.\\

The $i$-th coordinate of the right-hand side vector $\b$ is defined by 
$$b_i=1+\dfrac{1}{2}(w_i-l_i),$$
where $w_i$ (resp. $l_i$) is the total number of wins (resp. losses) accumulated by $P_i$. Therefore, Colley method's does not take into account match scores.

\begin{example}
We consider again the same tournament as the one used to illustrate the Borda method. In this situation, the Colley matrix $C$ and the vectors $\b$ and $\r$  are:
$$C=
\begin{pmatrix}
14&-2&-3&-7\\
-2&11&-3&-4\\
-3&-3&11&-3\\
-7&-4&-3&16
\end{pmatrix},
\quad
\b=\begin{pmatrix}
2\\1/2\\3/2\\0
\end{pmatrix}, \quad 
\r=\begin{pmatrix}
0.5509\ldots\\0.4574\ldots\\0.5357\ldots\\0.4558\ldots
\end{pmatrix}.\\
$$
Thus the Colley method gives the ranking $P_1 \succ P_3 \succ P_2 \succ P_4$.
\end{example}

\subsection{The Markov method.}
This method has been  studied extensively for ranking objects and has a long history \cite{sinn2022landau}. This method corresponds to Google's PageRank algorithm. We recall briefly here how it works in the context of a tournament.

Let us imagine that a fan changes his favorite player according to the results of matches. Initially, this fan
supports one player, for example $P_1$. Then, the supporter looks at the results of player $P_1$ and chooses to support
a player who has beaten $P_1$. The choice is made at random and the probability of choosing a player depends on the number of matches that player has won.
More precisely, the probability of choosing a player is proportional to the number of times that player has beaten $P_1$. The supporter now backs a new player, say $P_2$.
The process is repeated and if a player is undefeated then the supporter chooses a new player with equiprobability. If we continue this process, then counting the proportion of time the fan supports a player gives us an indication of the value of that player. We will now look at how to put this idea into practice.

We consider an $n \times n$ matrix $G$ whose entries are $G_{ji}=l_{ij}/n_i$, where $l_{ij}$ is the number of losses of $P_i$ against $P_j$ and $n_i$ is the total number of losses of $P_i$.
If the player $P_i$ is undefeated then all the coefficients of the $i$-th column are equal to $1/n$. This provides a stochastic matrix.\\

In the webpages context, we consider a directed graph where the vertices are webpages and the edges are hyperlinks from a page $P_i$ to a page $P_j$. Here, the vertices of the graph are the players and there is an edge from $P_i$ to $P_j$ if $P_i$ is defeated by $P_j$.
Then a matrix of the form $\alpha G + \frac{1-\alpha}{n} N_n$ is computed, where $N_n$ is the $n \times n$ matrix with all coefficients equal to $1$ and $\alpha$ is a parameter usually equal to $0.85$, see \cite{who1} and references therein. An eigenvector (with positive coordinates) $\r$  associated to the eigenvalue 1  of $R$ is computed. The coordinates $r_i$ of $\r$ give the rating of $P_i$ and then the  ranking between the players is obtained.\\

The PageRank algorithm has been studied in several articles. For example in \cite{Newlink}, the authors have worked on the impact of a new link and in \cite{perturbation_hyper_link} the stability of this algorithm is analyzed. However, the inversion paradox has not been studied.


\begin{example}
Using again the same tournament as that used to illustrate the Borda method, the matrix obtained when using the Markov method is as follows 
$$G=\begin{pmatrix}
0&1/5&1/4&5/8\\
1/5&0&1/4&2/8\\
2/5&2/5&0&1/8\\
2/5&2/5&2/4&0
\end{pmatrix}.
$$
Furthermore, an eigenvector associated to the eigenvalue $1$ of $\alpha G+\frac{1-\alpha}{4}N_4$ is
$$\r=\begin{pmatrix}
\frac{5(15 \alpha^3+85\alpha^2+172\alpha+160)}{16 \left(3 \alpha^{2}+20 \alpha +25\right) \left(\alpha +2\right)}
\\
-\frac{5 \left(3 \alpha^{3}-19 \alpha^{2}-112 \alpha -160\right)}{16 \left(3 \alpha^{2}+20 \alpha +25\right) \left(\alpha +2\right)}\\
\frac{\alpha +8}{4 \alpha +8}\\
1
\end{pmatrix}.$$
Furthermore, for all $\alpha \in [0,1]$, we get $r_4>r_1>r_3>r_2$. Thus the Markov method gives the ranking $P_4\succ P_1 \succ P_3 \succ P_2$.\\
\end{example}

\begin{remark}
In the previous examples, the tournament was the same but the four rankings are different:\\
- with the Borda method: $P_1 \succ P_4 \succ P_3 \succ P_2$.\\
- with the Massey method: $P_1 \succ P_3 \succ P_4 \succ P_2$. \\
- with the Colley method: $P_1 \succ P_3 \succ P_2 \succ P_4$.\\
- with the Markov method: $P_4\succ P_1 \succ P_3 \succ P_2$.
\end{remark}

\subsection{Other methods.}
There exists other methods for ranking players, see e.g., \cite{Stefani,BozzoVidoniFranceschet,RandomWalker,Dynamic_Bradley-Terry,ChartierKreutzerLangvillePedings,FranceschetBozzoVidoni,Chebotarev-Shamis,MinimumViolation,DevlinTreloar,Stob, Chebotarev,GonzDiaz-Hendrickx-Lohmann,Brink,Bertran,Ismail}. In this article, we have chosen to give a general theorem and then to show that, as a corollary, the inversion occurs with some classical methods. We have chosen to study the Borda, Massey, and Colley methods in particular because they are used in practice and, moreover, the proofs of the inversion paradox for these methods use only elementary linear algebra.

\section{Inversion with Condorcet tournaments.}\label{sec:2}
\subsection{A general result.}
When we have the perfect tournament $\mathcal{T}_{P_{\sigma(1)} \succ P_{\sigma(2)} \succ \cdots \succ P_{\sigma(n)}}$, only the following ranking is natural: $P_{\sigma(1)} \succ P_{\sigma(2)} \succ \cdots \succ P_{\sigma(n)}.$
This leads to the following definition.

\begin{definition}
We say that a ranking method $f$ is \emph{natural} if it satisfies
$$f(\mathcal{T}_{P_{\sigma(1)} \succ P_{\sigma(2)} \succ \cdots \succ P_{\sigma(n)}})=P_{\sigma(1)} \succ P_{\sigma(2)} \succ \cdots \succ P_{\sigma(n)}.$$
\end{definition}

By giving the same number of wins and losses to each player, a Condorcet tournament seems to give no information to rank the players. Thus, it is reasonable to assume that the presence or absence of a Condorcet tournament should not affect the ranking. This leads to the definition:

\begin{definition}
Let $\mathcal{T}$ be the union of different perfect tournaments and  $\mathcal{C}$ be a Condorcet tournament. We say that a ranking method $f$ can be \emph{reduced by Condorcet tournaments}, if for all $k \in \NN$, we have $$f(\mathcal{T}+ k\mathcal{C})=f(\mathcal{T}).$$
\end{definition}

This property has already been used to get strong results, see  for example  \cite{Bouyssou}, where it is called ``independence of circuits", or see \cite{Balinski_Laraki} where  the authors write ``the method cancels properly". \\

The following definition concerns tournaments of the form $\mathcal{T}_1+k\mathcal{T}_2$. When $k$ is sufficiently large and  $f(\mathcal{T}_2)$ is a strict total order, it seems reasonable to expect the ranking of players over this kind of long tournaments to be identical to the ranking $f(\mathcal{T}_2)$. This leads to the following definition.

\begin{definition}
Let $\mathcal{T}_1$ and $\mathcal{T}_2$ be the union 
of perfect tournaments with $f(\mathcal T_2)$ a strict total order.
We say that a ranking method $f$ satisfies 
\emph{the long tournament property}, 
if $$f(\mathcal{T}_1 + k \mathcal{T}_2)=f(\mathcal{T}_2),$$ 
when $k \in \NN$ is large enough.
\end{definition}
Now, we formally define what an inversion paradox is.

\begin{definition}
We say that  the \emph{inversion paradox} occurs 
with the ranking method $f$ if there exists a tournament $\mathcal T$ and a permutation $\sigma$
such that the following 
 conditions are satisfied:\\
1) $\mathcal{T}$ is a tournament with $n$ players $P_1$, \ldots, $P_n$ such that $$f(\mathcal{T})=P_{\sigma(1)} \succ P_{\sigma(2)} \succ \cdots \succ P_{\sigma(n)}.$$
2) If $\mathcal{T}'$ is the tournament where we consider 
all the results in $\mathcal{T}$ except the ones where $P_{\sigma(n)}$ appears (in other words, $\mathcal{T}'$ is the tournament $\mathcal{T}$ 
where player $P_{\sigma(n)}$ is deleted) then:
$$f(\mathcal{T}')=P_{\sigma(n-1)} \succ P_{\sigma(n-2)} \succ \cdots \succ P_{\sigma(2)} \succ P_{\sigma(1)}.$$
\end{definition}

\begin{remark}\label{rem:prop_prime}
We have introduced the notation $\mathcal{T}'$. We can remark that we have the following straightforward properties:
\begin{enumerate}
\item If $ \mathcal{T}_1$ and $\mathcal{T}_2$ are two tournaments and $\mathcal{T}'_1$, $\mathcal{T}'_2$ are the corresponding tournaments where the same player is deleted, then   $(\mathcal{T}_1+\mathcal{T}_2)'=\mathcal{T}_1'+\mathcal{T}_2'$.
\item If the player $P_n$ is deleted then $(\mathcal{T}_{P_1 \succ P_2 \succ \cdots \succ P_n})'= \mathcal{T}_{P_1 \succ P_2 \succ \cdots \succ P_{n-1}}$.
\end{enumerate}
\end{remark}

Furthermore, we have the following lemma.
\begin{lemma}\label{lem:condorcet_prime}
With the previous notations, if $P_n$ is deleted, we have 
$$(\mathcal{C}_{P_n \succ P_{n-1} \succ \cdots \succ P_1}\big)'= \mathcal{T}_{P_{n-1} \succ \cdots \succ P_1} + \mathcal{C}_{P_{n-1} \succ \cdots \succ P_1}.$$
\end{lemma}
\begin{proof}
We have
\begin{eqnarray*}
\mathcal{C}_{P_n \succ P_{n-1} \succ \cdots \succ P_1}&=&\mathcal{T}_{P_n \succ P_{n-1} \succ \cdots \succ P_1} + \mathcal{T}_{P_{n-1} \succ P_{n-2} \succ \cdots \succ P_1 \succ P_n} \\
&&+ \mathcal{T}_{P_{n-2} \succ P_{n-3} \succ \cdots \succ P_1 \succ P_n \succ P_{n-1}} + \cdots +\mathcal{T}_{P_1 \succ P_{n} \succ \cdots \succ P_2}.
\end{eqnarray*}
Thus if $P_n$ is deleted we get
\begin{eqnarray*}
\big(\mathcal{C}_{P_n \succ P_{n-1} \succ \cdots \succ P_1}\big)'&=&\mathcal{T}_{P_{n-1} \succ \cdots \succ P_1} + \mathcal{T}_{P_{n-1} \succ \cdots \succ P_1} \\
&&+ \mathcal{T}_{P_{n-2} \succ P_{n-3} \succ \cdots \succ P_1 \succ P_{n-1}} + \cdots +\mathcal{T}_{P_{1} \succ P_{n-1} \succ \cdots \succ P_2}\\
&=& \mathcal{T}_{P_{n-1} \succ \cdots \succ P_1} + \mathcal{C}_{P_{n-1} \succ \cdots \succ P_1}.
\end{eqnarray*}
\end{proof}

Our main result states that a ranking method satisfying the previous definitions  gives an inversion paradox.

\begin{theorem}\label{thm:inv}
If a natural ranking method can be reduced by  Condorcet tournaments and satisfies the long tournament property then it suffers from the inversion paradox.
\end{theorem}

\begin{proof}
Let $f$ be a ranking method which satisfies the hypotheses of the theorem. As $f$ is a natural ranking method we have 
$$f(\mathcal{T}_{P_1 \succ P_2 \succ \cdots \succ P_n})=P_1 \succ P_2 \succ \cdots \succ P_n.$$
We consider $\mathcal{T}_k=\mathcal{T}_{P_1 \succ P_2 \succ \cdots \succ P_n}+k \mathcal{C}_{P_n \succ P_{n-1} \succ \cdots \succ P_1}$.
As $f$ can be reduced by Condorcet tournaments, we have for all $k \in \NN$
$$f(\mathcal{T}_k)=P_1 \succ P_2 \succ \cdots \succ P_n.$$

Now, we consider the tournament $\mathcal{T}_k'$  where the player $P_n$ is deleted.  Thanks to Remark~\ref{rem:prop_prime} and Lemma~\ref{lem:condorcet_prime}, we have
\begin{eqnarray*}
\mathcal{T}_k'&=& \big(\mathcal{T}_{P_1 \succ P_2 \succ \cdots \succ P_n}+k \mathcal{C}_{P_n \succ P_{n-1} \succ \cdots \succ P_1}\big)'\\
&=& \big(\mathcal{T}_{P_1 \succ P_2 \succ \cdots \succ P_n}\big)'+k \big(\mathcal{C}_{P_n \succ P_{n-1} \succ \cdots \succ P_1}\big)'\\
&=&\mathcal{T}_{P_1 \succ P_2 \succ \cdots \succ P_{n-1}}+ k\big( \mathcal{T}_{P_{n-1} \succ \cdots \succ P_1} + \mathcal{C}_ {P_{n-1} \succ \cdots \succ P_1}\big)\\
&=& \mathcal{T}_{P_1 \succ P_2 \succ \cdots \succ P_{n-1}}+k \mathcal{T}_ {P_{n-1} \succ \cdots \succ P_1} +k \mathcal{C}_ {P_{n-1} \succ \cdots \succ P_1}.
\end{eqnarray*}
Thus, as $f$ can be reduced by Condorcet tournaments we have
$$f(\mathcal{T}_k')=f(\mathcal{T}_{P_1 \succ P_2 \succ \cdots \succ P_{n-1}}+k \mathcal{T}_ {P_{n-1} \succ \cdots \succ P_1}).$$
Furthermore, as $f$ satisfies the long tournament property, there exists a sufficiently large $k$ such that 
$$f(\mathcal{T}_k')=f(\mathcal{T}_ {P_{n-1} \succ \cdots \succ P_1}).$$\\
At last, as $f$ is a natural ranking method we get
$$f(\mathcal{T}_k')=P_{n-1} \succ \cdots \succ P_1.$$
Thus the inversion paradox occurs.
\end{proof}

\begin{remark}
As mentionned in the introduction, the previous theorem can be stated as an impossibility theorem. Indeed, we have proved that a natural ranking method cannot simultaneously satisfy the long tournament property, reducibility by Condorcet tournaments and avoid the inversion paradox.
\end{remark}

In the following we will show that the methods of Borda, Massey, and Colley verify the hypotheses of the theorem. Consequently, the inversion paradox  can appear and we will illustrate it with the same tournament $\mathcal{X}=\mathcal{T}_{P_1 \succ P_2 \succ \cdots \succ P_5} + 2 \mathcal{C}_{P_5 \succ P_{4} \succ \cdots \succ P_1}$.


\subsection{Application to the Borda method.}
We are going to prove the following result:
\begin{corollary}
The inversion paradox occurs with the Borda method.
\end{corollary}
In order to obtain this corollary, we just have to prove the following proposition. 

\begin{proposition}\label{prop:borda}The Borda method is i) natural, ii) can be reduced by  Condorcet tournaments, iii) satisfies the long tournament property.
\end{proposition}

\begin{proof}
\emph{i)} If we consider the perfect tournament $\mathcal{T}_{ P_{\sigma(1)} \succ \cdots \succ P_{\sigma(n)} }$, we have already noticed that in this situation $P_{\sigma(1)}$ has $n-1$ wins, $P_{\sigma(2)}$ has $n-2$ wins, etc. This gives the ranking: $P_{\sigma(1)} \succ \cdots \succ P_{\sigma(n)}$.\\
\emph{ii)} If we add a Condorcet tournament to a tournament with $n$ players,  each player gets  $n(n-1)/2$ wins. Thus, we add the same numbers of wins to each player and then the final ranking is not modified.\\
\emph{iii)} We denote by $w_{1,i}$ (respectively $w_{2,i}$)  the number of wins of player $P_i$ in the tournament $\mathcal{T}_1$ (respectively $\mathcal{T}_2$). We suppose that the Borda method applied to $\mathcal{T}_2$ gives a strict total ordering. Thus we suppose that $w_{2,\sigma(1)} > w_{2, \sigma(2)}> \cdots > w_{2,\sigma(n)}$. Therefore, we have 
\begin{small}$$ w_{1,\sigma(i)} + k w_{2,\sigma(i)} >  w_{1,\sigma(j)} + k w_{2,\sigma(j)}  \iff k(w_{2,\sigma(i)}-w_{2,\sigma(j)}) > w_{1,\sigma(j)}-w_{1,\sigma(i)}.$$
\end{small}
When $i<j$, we have  $w_{2,\sigma(i)}-w_{2,\sigma(j)}>0$, then the second inequality is satisfied if $k$ is big enough. Thus, the first inequality is satisfied when $k$ is big enough.
As $ w_{1,\sigma(i)} + k w_{2,\sigma(i)}$ corresponds to the number of wins of player $P_{\sigma(i)}$ in the tournament $\mathcal{T}_1+ k \mathcal{T}_2$, we deduce that the Borda method satisfies the long tournament property.
\end{proof}

\begin{example} 
We give an example of an inversion paradox 
 with the Borda method when we have $n=5$ players. Our construction follows the one given in the proof of Theorem~\ref{thm:inv}.
We set
$$\mathcal{X}= \mathcal{T}_{P_1 \succ P_2 \succ \cdots \succ P_5} + 2 \mathcal{C}_{P_5 \succ P_4 \succ \cdots \succ P_1}.$$
We denote by $w_i$ the number of wins of player $P_i$ in the tournament $\mathcal{X}$. We have $w_1=24$, $w_2=23$, $w_3=22$, $w_4=21$ and $w_5=20$. Thus, the Borda method applied to $\mathcal{X}$ gives the final ranking: $P_1 \succ P_2 \succ \cdots \succ P_5$.\\
Now, if $P_5$ is deleted, the  tournament $\mathcal{X}'$ is 
$$\mathcal{X}'=\mathcal{T}_{P_1 \succ P_2 \succ P_3 \succ P_4} +2 \mathcal{T}_{P_4 \succ P_3 \succ P_2 \succ P_1} +2 \mathcal{C}_{P_4 \succ P_3 \succ P_2 \succ P_1}.$$
We denote by $w'_i$ the number of wins of player $P_i$ in the tournament $\mathcal{X}'$. We have $w'_1=15$, $w'_2=16$, $w'_3=17$, $w'_4=18$. Thus, the Borda method applied to $\mathcal{X}'$ gives the final ranking: $P_4 \succ P_3 \succ P_2  \succ P_1$.\\

If we consider the tournament where we add just one Condorcet tournament $\mathcal{C}_{P_5 \succ P_4 \succ \cdots \succ P_1}$ to  the perfect tournament $\mathcal{T}_{P_1 \succ P_2 \succ \cdots \succ P_5}$, then $w'_1=w'_2=w'_3=w'_4=9$. Thus, in this situation, if $P_5$ is deleted then we get $P_1=P_2=P_3=P_4$. The order is modified but we do not get an inversion paradox.\\
\end{example}

\subsection{Application to the Massey method.}
In the following, we denote by $f_{M}$ 
the Massey  method and 
 we are going to prove:
\begin{corollary}
The inversion paradox occurs with the Massey method.
\end{corollary}
In order to prove this corollary, as before we are going to show 
that the Massey method is natural, can be reduced by Condorcet tournaments 
and satisfies the long tournament property. Let us begin with some remarks and notations. 
The Massey matrix  $M_{t,n}$  associated 
with a tournament $\mathcal{T}$ made up of $t$  perfect
tournaments with $n$ players is the following:

$$M_{t,n}=
\left(
\begin{array}{ccccc}
t(n-1) & -t &  \ldots & -t & -t\\
-t & t(n-1) & \ldots & -t & -t\\
\vdots & \vdots &  \ddots &\vdots &\vdots \\
-t&-t &\ldots &t(n-1)  &-t\\
-t &-t &\ldots &-t &  t(n-1)
\end{array}
\right).
$$

\noindent The \textsl{cumulative point differential vector of $\mathcal{T}$}, denoted by $\p$  or by $\p(\mathcal{T})$ to be more precise, is given by
$$
~~~~~~~~~~~~~~
\p=\left(
\begin{array}{c}
p_1\\
p_2\\
\vdots\\
p_{n-1}\\
p_n
\end{array}
\right)
:=\left(
\begin{array}{c}
w_1-l_1\\
w_2-l_2\\
\vdots\\
w_{n-1}-l_{n-1}\\
w_n-l_n
\end{array}
\right)
$$
where $w_i$ (resp. $l_i$) is the total number of wins 
(resp. losses) accumulated by player
$P_i$. Indeed, as we consider perfect tournaments, the results of all matches are $1-0$ or \mbox{$0-1$} and $p_i=w_i-l_i$. \\
We  remark that $ \sum_{i=1}^np_i=0$. Thus, $p_n=-(p_1+\cdots+p_{n-1})$.

\noindent The ranking $f_{M}(\mathcal{T})$ is given by the vector $\r^M=(r_1 \cdots r_n)^T$
such that $\ovl{M}_{t,n}\: \r^M= \ovl{\p}$
 with
$$\ovl M_{t,n}=
\left(
\begin{array}{cccc|c}
t(n-1) & -t &  \ldots & -t & -t\\
-t & t(n-1) & \ldots & -t & -t\\
\vdots & \vdots &  \ddots &\vdots &\vdots \\
-t&-t &\ldots &t(n-1)  &-t\\ \hline
1&1 &\ldots &1 &  1
\end{array}
\right)
\:
~~~~~~~{\rm and}~~~~~~~
\ovl \p=\left(
\begin{array}{c}
p_1\\
p_2\\
\vdots\\
p_{n-1}\\
\hline
0
\end{array}
\right).
$$

Let us recall that the \textit{Massey rating vector}
$\r^{M}$ (or $\r^{M}(\mathcal{T})$ to be more precise)
 has been chosen in order to have $\disp \sum_{i=1}^n\r^{M}_i=0$.\\

\begin{proposition}\label{prop-massey}
With the previous notations, the inverse of $\ovl M_{t,n}$
is~: 
$$
\left( \ovl M_{t,n} \right)^{-1}=
\frac{1}{nt}
\left(
\begin{array}{ccc|c}
 &  &     & t\\
 &    I_{n-1} & &\vdots \\
& &  &t\\ \hline
-1 &\ldots &-1 &  t
\end{array}
\right)
$$
where $I_{n-1}$ is the $(n-1)\times (n-1)$-identity matrix.
\end{proposition}

\begin{proof}
The proof comes from a direct computation.
\end{proof}

\begin{corollary}\label{cor-massey}
We consider tournaments with $n$ players.

\begin{enumerate}
\item[i.] If $\mathcal{T}$ is a  tournament made of $t$ perfect tournaments, then
\begin{equation}
\tag{1}
\r^M(\mathcal{T})=\frac{1}{nt}\:\p(\mathcal{T}). \label{EqM}
\end{equation}
\item[ii.] \label{item2} Let $\mathcal{T}$ and $\mathcal{T}'$ be two tournaments made of, 
$t$ and $t'$ perfect tournaments. 
If $\p(\mathcal{T})=\p(\mathcal{T}')$, then 
$\disp \r^{M}(\mathcal{T}')=\frac{t}{t'}\:\r^{M}(\mathcal{T})$ and $f_{M}(\mathcal{T})=f_{M}(\mathcal{T}')$.
\end{enumerate} 
\end{corollary}

\begin{proof}
\begin{enumerate}
\item[i.] By Proposition \ref{prop-massey}, and thanks to the equality $p_n=-(p_1+\cdots+p_{n-1})$ we have 
\begin{eqnarray*}
\r^{M}(\mathcal{T})&=&\ovl M_{t,n}^{-1}\:\ovl \p=
\frac{1}{nt}\:
\left(
\begin{array}{cccc|c}
1 & 0 &  \ldots & 0 & t\\
0 & 1 & \ldots & 0 & t\\
\vdots & \vdots &  \ddots &\vdots &\vdots \\
0&0 &\ldots &1  &t\\ \hline
-1&-1 &\ldots &-1 &  t
\end{array}
\right)
\:
\left(
\begin{array}{c}
p_1\\
p_2\\
\vdots\\
p_{n-1}\\
\hline
0
\end{array}
\right)\\
&
=&
\frac{1}{nt}\:
\left(
\begin{array}{c}
p_1\\
p_2\\
\vdots\\
p_{n-1}\\
-(p_1+\cdots+p_{n-1})
\end{array}
\right)
=\frac{1}{nt}\:\p(\mathcal{T}).
\end{eqnarray*}
\item[ii.] It is a straightforward application of (\ref{EqM}).
%
\end{enumerate} 
\end{proof}

\noindent The following assertion is a direct consequence
of the equation  (\ref{EqM}) in Corollary~\ref{cor-massey}
but we emphasize it because it is the main tool in the sequel.

\begin{corollary}\label{cor_massey-strict-order}
If $\mathcal{T}$ is a  tournament made of perfect tournaments, then\\
\mbox{$\p(\mathcal{T})=(p_1,p_2,\ldots,p_n)^{T}$} and 
$\r^{M}(\mathcal{T})=(r_1,r_2,\ldots,r_n)^{T}$ 
have the same variations~:
$$
\forall\: 1\leq i<j\leq n\:,~~~
r_i<r_j \Longleftrightarrow p_i <p_j.$$
\end{corollary}

\noindent
We can now prove that the Massey method satisfies the hypotheses of Theorem~\ref{thm:inv}.

\begin{proposition}\label{prop:massey_3prop}
The Massey method is i) natural, ii) can be reduced by  Condorcet tournaments, iii) satisfies the long tournament property.
\end{proposition}

\begin{proof}
\noindent $i)$ Let $\mathcal{T}=P_{\sigma(1)} \succ P_{\sigma(2)} \succ \cdots \succ P_{\sigma(n)}$ be a perfect tournament and let $p(\mathcal{T})=(p_1,\ldots,p_n)^T$ be the associated cumulative point differential vector. We have 
$$p_{\sigma(1)}=n-1> p_{\sigma(2)} =n-3> \cdots > p_{\sigma(n)}=-(n-1).$$
 By Corollary~\ref{cor_massey-strict-order}, we have $r_i <r_j$ if and only if $p_i<p_j$, thus, here we get \linebreak $r_{\sigma(1)} > \cdots > r_{\sigma(n)}$. Therefore, the Massey method is natural. 
\smallbreak

\noindent $ii)$ Let $\mathcal{T}$ be the union of $t$ perfect tournaments 
between $n$ players and let $\mathcal{C}$ be a Condorcet tournament.
As $\mathcal{C}$ is a Condorcet tournament,
we  remark that $\p(\mathcal{T})=\p(\mathcal{T}+k\mathcal{C})$.
So, by Corollary \ref{cor_massey-strict-order},
$\r^{M}(\mathcal{T})$ and $\r^{M}(\mathcal{T}+k\mathcal{C})$
have the same variations, since they have the same variations 
as $\p(\mathcal{T})$. This means that $f_{M}(\cT)=f_{M}(\cT+k\cC)$.
\smallbreak

\noindent $iii)$
Let $\mathcal{T}_1$ be the union of $t_1$ perfect tournaments 
between $n$ players and $\mathcal{T}_2$ the union of $t_2$ 
perfect tournaments between these players. 
The cumulative  point differential vectors verify
 $\p(\mathcal{T}_1 + k \mathcal{T}_2)=
\p(\mathcal{T}_1) + k \p(\mathcal{T}_2) $ 
and for the rating vectors,
we have by Corollary \ref{cor-massey}, 
\begin{eqnarray*}
\r^M (\mathcal{T}_1 + k \mathcal{T}_2)&=&
\frac{1}{n(t_1+kt_2)}\p(\mathcal{T}_1 + k \mathcal{T}_2)=
\frac{1}{n(t_1+kt_2)}\big(\p(\mathcal{T}_1) + k \p(\mathcal{T}_2)\big)\\
&=&
\frac{1}{n\left(\disp  \frac{t_1}{k}+t_2\right)}
\left(\frac{1}{k}\p(\mathcal{T}_1)+\p(\mathcal{T}_2)\right).
\end{eqnarray*}
It follows that 
$$\disp \lim_{k\to \infty} \r^{M} (\mathcal{T}_1 + k \mathcal{T}_2)
=\frac{1}{nt_2}\p(\mathcal{T}_2)=\r^{M} (\cT_2).$$

\noindent If $f(\mathcal{T}_2)$ is a strict ordering then the coordinates of $\r^M(\mathcal{T}_2)$ are also strictly linearly ordered. Thus, the previous equality shows that when $k$ is big enough $\r^M(\mathcal{T}_1 + k \mathcal{T}_2)$ and $\r^M(\mathcal{T}_2)$ give the same ranking. Therefore, the Massey method satisfies the long tournament property.
%
\end{proof}

\begin{example}
 We give an example of an inversion paradox 
when we have $n=5$ players. 
As with the Borda method, we consider
$\mathcal{X}= \mathcal{T}_{P_1 \succ P_2 \succ \cdots \succ P_5} + 2 \mathcal{C}_{P_5 \succ P_4 \succ \cdots \succ P_1}$.
The Massey matrix associated to the perfect tournament $\mathcal{T}_{P_1 \succ P_2 \succ \cdots \succ P_5}$ is $M_{1,5}$. The Condorcet tournament $\mathcal{C}_{P_5 \succ P_4 \succ \cdots \succ P_1}$ is the union of five perfect tournaments, thus the Massey matrix associated to  $\mathcal{C}_{P_5 \succ P_4 \succ \cdots \succ P_1}$ is $M_{5,5}$. 
We denote by $M$ the Massey matrix associated to $\mathcal{X}$ and we have 
$$
M=M_{1,5}+2M_{5,5}=M_{11,5}=\begin{pmatrix}
44&-11&-11&-11&-11\\
-11&44&-11&-11&-11\\
-11&-11&44&-11&-11\\
-11&-11&-11&44&-11\\
-11&-11&-11&-11&44
\end{pmatrix}.
$$
The cumulative differential point vector $p$ associated to the perfect tournament  $\mathcal{T}_{P_1 \succ P_2 \succ \cdots \succ P_5}$ is $(4,2,0,-2,-4)^T$ and the one associated to $\mathcal{C}_{P_5 \succ P_4 \succ \ldots \succ P_1}$ is $(0,0,0,0,0)^T$. Thus the cumulative differential point vector $\p$ associated to  $\mathcal{X}$ is 
$$\p=(4,2,0,-2,-4)^T.$$
By (\ref{EqM}) in Corollary \ref{cor-massey},
the solution of $\overline{M}\r=\overline{\p}$ is 
$$\r=\Big(\dfrac{4}{55}, \dfrac{2}{55},0,-\dfrac{2}{55}, -\dfrac{4}{55}\Big)^T.$$
This gives the ranking: $P_1 \succ P_2 \succ P_3 \succ P_4 \succ P_5$.\\

Now, if $P_5$ is deleted then the tournament $\mathcal{X}'$ is 
$$\mathcal{T}_{P_1 \succ P_2 \succ P_3 \succ P_4} +2 \mathcal{T}_{P_4 \succ P_3 \succ P_2 \succ P_1} + 2 \mathcal{C}_{P_4 \succ P_3 \succ P_2 \succ P_1}.$$
 The Massey matrix associated to $\mathcal{X}'$ is 
$$M'= M_{1,4}+2M_{1,4}+2M_{4,4}=M_{11,4}=
\begin{pmatrix}
33&-11&-11&-11\\
-11&33&-11&-11\\
-11&-11&33&-11\\
-11&-11&-11&33
\end{pmatrix}.
$$
The cumulative differential point vector associated to 
the perfect tournament  $\mathcal{T}_{P_1 \succ P_2 \succ P_3 \succ P_4}$ 
is $(3,1,-1,-3)^T$, the one associated to the perfect tournament 
$\mathcal{T}_{P_4 \succ P_3 \succ P_2 \succ P_1}$ is $(-3,-1,1,3)^T$, 
and the one associated to $\mathcal{C}_{P_4 \succ P_3 \succ P_2 \succ P_1}$ is $(0,0,0,0)^T$. 
Thus the cumulative differential point vector $\p'$ associated 
to  $\mathcal{X}'$ is 
$$\p'=(-3,-1,1,3)^T.$$
By (\ref{EqM}) in Corollary \ref{cor-massey}, the solution
 of $\overline{M'}\r'=\overline{\p}'$ is 
$$\r'=\Big(-\dfrac{3}{44},-\dfrac{1}{44},\dfrac{1}{44},
\dfrac{3}{44}\Big)^T.$$
This gives the ranking: $P_4 \succ P_3 \succ P_2 \succ P_1$. Thus an inversion paradox occurs when we consider the tournament $\mathcal{X}'=\mathcal{T}_{P_1 \succ P_2 \succ \cdots \succ P_5} + 2 \mathcal{C}_{P_5 \succ P_4 \succ \cdots \succ P_1}$.\\

If we consider the tournament where we add just one Condorcet tournament $\mathcal{C}_{P_5 \succ P_4 \succ \cdots \succ P_1}$ to  the perfect tournament $\mathcal{T}_{P_1 \succ P_2 \succ \cdots \succ P_5}$, then the vector $\p$ is still $(4,2,0,-2,-4)$ but $\p'$ becomes $(0,0,0,0)^T$. Thus, in this situation, deleting $P_5$  implies $P_1=P_2=P_3=P_4$. The order is modified but we do not get an inversion paradox.\\
\end{example}

\subsection{Application to the Colley method.}


The calculations we have just made using Massey method 
allow us to get a similar conclusion for Colley method.
Actually, while the rating vector $\r^{M}$ of the Massey method
is given by $M\r^{M}=\p$, 
 the rating vector $\r^{C}$ of the Colley method
 is given by $C\r^C=\b$ where $C=M+2I$ and $\b=\mathbf{e} +\frac{1}{2}\p$ with $\mathbf{e}=(1,\ldots,1)^T$.
 In particular, il we consider a tournament made up of
 the union of $t$ perfect tournaments between 
$n$ players, $M_{t,n}\mathbf{e}=0$ or 
$C_{t,n}\mathbf{e}=2 \mathbf{e}$ and if we set
$\t r^{C}:=2\r^{C}-\mathbf{e}$, we get
\begin{eqnarray*}
C_{t,n}\r^{C}=\b &\Longleftrightarrow &
C_{t,n}(2 \r^{C}-\mathbf{e})=2\b -C_{t,n}(\mathbf{e})
 \Longleftrightarrow
C_{t,n}(\t \r^{C})=2(\b-\mathbf{e})=\p.
\end{eqnarray*}
Therefore, since $\r^{C}$ and $\t \r^{C}$  give the same ranking, 
we have to find  $\t \r^{C}$ which is the solution
of  $C_{t,n}\t \r^{C}=\p$ or, equivalently,
$(M_{t,n}+2I_n)\t \r^{C}=\p$. By Corollary~\ref{cor-massey},
we get $(M_{t,n}+2I)\p=nt\p+2\p$ and the solution
\begin{equation}
\tag{2}
\t \r^{C}(\mathcal{T})=\frac{1}{nt+2}\:\p(\mathcal{T}). \label{EqC}
\end{equation}
Consequently, using the same reasoning as for Massey method, 
we obtain the following conclusion.

\begin{proposition}
The Colley method is i) natural, ii) can be reduced by  Condorcet tournaments, iii)  satisfies the long tournament property.
\end{proposition}

\begin{corollary}
The inversion paradox occurs with the Colley method.
\end{corollary}


\begin{example} 
 As before, we consider
$\mathcal{X}= \mathcal{T}_{P_1 \succ P_2 \succ \cdots \succ P_5} + 2 \mathcal{C}_{P_5 \succ P_4 \succ \cdots \succ P_1}$. The Colley matrix and the vector $\b$ associated to $\mathcal{X}$ are
$$
C=\begin{pmatrix}
46&-11&-11&-11&-11\\
-11&46&-11&-11&-11\\
-11&-11&46&-11&-11\\
-11&-11&-11&46&-11\\
-11&-11&-11&-11&46
\end{pmatrix}
\textrm{ and \,\,}
\b=\begin{pmatrix}
3\\2\\1\\0\\-1
\end{pmatrix}.
$$

By (\ref{EqC}), the
solution of $C\r=\p$ is $\r=\Big(\dfrac{61}{114}, \dfrac{59}{114},\dfrac{57}{114},\dfrac{55}{114},\dfrac{53}{114}\Big)^T$ and this gives the ranking: $P_1 \succ P_2 \succ P_3 \succ P_4 \succ P_5$.\\

Now, if $P_5$ is deleted then we have the tournament $\mathcal{X}'$.  The Colley matrix  and the vector $\b'$ associated to $\mathcal{X}'$ are
$$C'= 
\begin{pmatrix}
35&-11&-11&-11\\
-11&35&-11&-11\\
-11&-11&35&-11\\
-11&-11&-11&35
\end{pmatrix}
\textrm{ and  \,\,}
\b'= \begin{pmatrix}
-1/2 \\ \hphantom{-}1/2 \\ \hphantom{-}3/2 \\ \hphantom{-}5/2
\end{pmatrix}.
$$

By (\ref{EqC}), the
solution of $C'\r'=\b'$ is $\r'=\Big(\dfrac{43}{92},\dfrac{45}{92},\dfrac{47}{92},\dfrac{49}{92}\Big)^T$ and this gives the ranking: $P_4 \succ P_3 \succ P_2 \succ P_1$. Thus an inversion paradox occurs.\\

If we consider the tournament where we add just one Condorcet tournament \linebreak $\mathcal{C}_{P_5 \succ P_4 \succ \cdots \succ P_1}$ to the perfect tournament $\mathcal{T}_{P_1 \succ P_2 \succ \cdots \succ P_5}$, then we obtain a situation similar to that obtained with Massey's method. That is, if $P_5$ is deleted, we obtain $P_1=P_2=P_3=P_4$. The order is changed but we do not get an inversion paradox. 
\end{example}

\subsection{A remark about the Markov method.}
With the Markov method applied to the tournaments  $\mathcal{X}= \mathcal{T}_{P_1 \succ P_2 \succ \cdots \succ P_5} + 2 \mathcal{C}_{P_5 \succ P_4 \succ \cdots \succ P_1}$,  an inversion paradox occurs. However, this paradox is not the one  given by  Theorem~\ref{thm:inv}.
The matrix $G$ associated to this tournament when we use the Markov method is
$$G=\begin{pmatrix}
0 & 3/21 & 5/22 & 7/23 & 9/24 
\\
 8/20 & 0 & 3/22 & 5/23 & 7/24 
\\
 6/20 & 8/21 & 0 & 3/23 & 5/24
\\
 4/20 & 6/21 &  8/22& 0 & 3/24
\\
 2/20 & 4/21 & 6/22 & 8/23 & 0 
\end{pmatrix}.
$$
Indeed, the total number of losses of $P_1$ is 20 
and $P_1$ is defeated 8 times by $P_2$, 6 times 
by $P_3$, 4 times by $P_4$, and 2 times by $P_5$, and so on\ldots\\
 With a computer algebra system, we get the following  eigenvector  associated 
to the eigenvalue $1$ of \mbox{$\alpha G+\frac{1-\alpha}{5}N_5$}, where $N_5$ is the  $5 \times 5$ matrix with all coefficients equal to 1:
$$\r=
\begin{pmatrix}
\frac{12345 \alpha^{4}+137110 \alpha^{3}+609815 \alpha^{2}+1338210 \alpha +1275120}{8148 \alpha^{4}+99852 \alpha^{3}+484608 \alpha^{2}+1161672 \alpha +1275120}
\\
\frac{4319 \alpha^{4}+47866 \alpha^{3}+210539 \alpha^{2}+444346 \alpha +425040}{2716 \alpha^{4}+33284 \alpha^{3}+161536 \alpha^{2}+387224 \alpha +425040}
\\
\frac{13519 \alpha^{4}+147686 \alpha^{3}+623469 \alpha^{2}+1300266 \alpha +1275120}{8148 \alpha^{4}+99852 \alpha^{3}+484608 \alpha^{2}+1161672 \alpha +1275120}
\\
\frac{13519 \alpha^{4}+147686 \alpha^{3}+623469 \alpha^{2}+1300266 \alpha +1275120}{8148 \alpha^{4}+99852 \alpha^{3}+484608 \alpha^{2}+1161672 \alpha +1275120}
\\
\tiny{1}
\end{pmatrix}.
$$
Furthermore, if $1 \geq \alpha>2/9\approx 0.22$ then this gives the ranking:\\
$P_2\succ P_1 \succ P_3 \succ P_4 \succ P_5$, while the Markov method applied to $\mathcal{T}_{P_1 \succ P_2 \succ \cdots \succ P_5}$ gives the ranking $P_1 \succ P_2 \succ \cdots \succ P_5$. This means that the  Markov method  
cannot be reduced by Condorcet tournaments.\\

Now, we consider the tournament $\mathcal{X}'$ 
when $P_5$ is deleted. Then the  matrix used 
by the Markov method is
$$G'=\begin{pmatrix}
0& 3/17&5/16&7/15\\
8/18&0&3/16&5/15\\
6/18&8/17&0&3/15\\
4/18&6/17&8/16&0
\end{pmatrix}.
$$
This means for example that in $\mathcal{X}'$ 
the player $P_2$ loses 17 matches and among them 3 losses are against $P_1$.  

An eigenvector  associated 
to the eigenvalue $1$ of \mbox{$\alpha G+\frac{1-\alpha}{4}N_4$} is 
$$\r'=
\begin{pmatrix}
\frac{741 \alpha^{3}+4818 \alpha^{2}+11697 \alpha +12240}{895 \alpha^{3}+5665 \alpha^{2}+13160 \alpha +12240}\\
\frac{2108 \alpha^{3}+13923 \alpha^{2}+35445 \alpha +36720}{2685 \alpha^{3}+16995 \alpha^{2}+39480 \alpha +36720}\\
\frac{2064 \alpha^{3}+14528 \alpha^{2}+36864 \alpha +36720}{2685 \alpha^{3}+16995 \alpha^{2}+39480 \alpha +36720}\\
1
\end{pmatrix}.
$$
Furthermore, if $1 \geq \alpha> -\frac{531}{230} + \frac{\sqrt{444801}}{230}\approx 0.59$ then this gives the ranking:\\ $P_4\succ P_3 \succ P_1 \succ P_2$.\\

We remark that, if $1 \geq \alpha > 0.59$, then an inversion paradox occurs. Indeed,
 when $P_5$ is  deleted we obtain the opposite order 
between the other players. This example shows that Theorem~\ref{thm:inv} 
gives only a sufficient condition for the occurrence of an inversion paradox since Markov method is not reducible by Condorcet tournaments as we have seen before.




\section{Inversion with few matches.}
In the previous section, we have \mbox{obtained} an inversion paradox with tournaments of the following kind: $\mathcal{T}_k=\mathcal{T}_{P_1 \succ P_2 \succ \cdots \succ P_n}+k \mathcal{C}_{P_n \succ P_{n-1} \succ \cdots \succ P_1}$, where $k$ is big enough. We have seen in the examples that  an inversion paradox is obtained  for the Borda, Massey, Colley, and Markov methods, if $k=2$. In this situation, each player has $(n-1)(2n+1)$ matches. If we have five players, then each player must play 44 matches. This situation can appear if we consider the results of several matches over several seasons. However, this leads to the following question: Can we obtain an inversion paradox when the number of matches of each player is small?

To answer this question, we introduce some notations.

\begin{definition}
$\mathcal{R}_{n}^+$ is the tournament with $n$ matches 
between players $P_1$, \ldots, $P_n$, where the $i$-th 
match is $(P_i,P_{i+1}, 0,1)$, when $i=1, \ldots,n-1$, 
and the last match is $(P_n,P_1,0,1)$.\\
$\mathcal{S}_n^-$ is the tournament with $n-1$ matches between 
players $P_1$, \ldots $P_n$ where the $i$-th match is 
$(P_i,P_{i+1},1,0)$, where $i=1, \ldots, n-1$. \\
$\mathcal{S}_n^+$ is the same tournament as $\mathcal{S}_n^-$ but 
with the opposite results,
i. e., with 
$(P_i,P_{i+1},0,1)$, where $i=1, \ldots, n-1$.\\
For every $(k,\ell)\in \mathbb{N}\times \mathbb{N}\setminus \{(0,0)\}$,
we define the tournament \mbox{$\mathcal{Z}_{n,k,l}=k\cR_n^+ + \ell\cS_n^-$}.
\end{definition}

 We can represent $\mathcal{R}_n^+$ and $\mathcal{S}_n^-$
 in the following way:\\
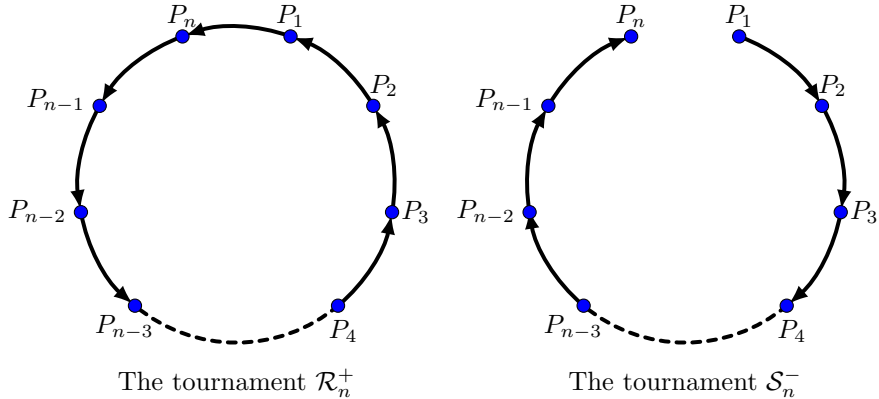
 \begin{tikzpicture}
[line cap=round,line join=round,>=latex,x=0.7cm,y=0.7cm]
\draw[->,line width=1.5pt] (3*cos{70},3*sin{70}) arc (70:109:3);
\draw[->,line width=1.5pt] (3*cos{30},3*sin{30}) arc (30:69:3);
\draw[->,line width=1.5pt] (3*cos{-10},3*sin{-10}) arc (-10:29:3);
\draw[->,line width=1.5pt] (3*cos{-50},3*sin{-50}) arc (-50:-11:3);
\draw[dashed, line width=1.5pt] (3*cos{-50},3*sin{-50}) 
arc (-50:-130:3);

\draw[->,line width=1.5pt] (3*cos{-170},3*sin{-170}) arc(-170:-131:3);
\draw[->,line width=1.5pt] (3*cos{-210},3*sin{-210}) arc(-210:-171:3);
\draw[->,line width=1.5pt] (3*cos{-250},3*sin{-250}) arc(-250:-211:3);

\draw[fill=blue] (3*cos{30},3*sin{30}) circle (2.5pt);
\draw[fill=blue] (3*cos{-10},3*sin{-10}) circle (2.5pt);
\draw[fill=blue] (3*cos{-50},3*sin{-50}) circle (2.5pt);
\draw[fill=blue] (3*cos{70},3*sin{70}) circle (2.5pt);
\draw[fill=blue] (3*cos{-130},3*sin{-130}) circle (2.5pt);
\draw[fill=blue] (3*cos{-170},3*sin{-170}) circle (2.5pt);
\draw[fill=blue] (3*cos{150},3*sin{150}) circle (2.5pt);
\draw[fill=blue] (3*cos{110},3*sin{110}) circle (2.5pt);

\draw (3*cos{70},3*sin{70}) node[above]{$P_1$};
\draw (3*cos{30}+0.2,3*sin{30}) node[above]{$P_2$};
\draw (3*cos{-10},3*sin{-10}) node[right]{$P_3$};
\draw (3*cos{-50}+0.1,3*sin{-50}-0.1) node[below]{$P_4$};
\draw (3*cos{-130}-0.2,3*sin{-130}) node[below]{$P_{n-3}$};
\draw (3*cos{-170}-0.1,3*sin{-170}) node[left]{$P_{n-2}$};
\draw (3*cos{150}-0.1,3*sin{150}+0.1) node[left]{$P_{n-1}$};
\draw (3*cos{110},3*sin{110}) node[above]{$P_{n}$};

\draw (0,-4.2) node[above]{The tournament $\mathcal R_n^+$};
\end{tikzpicture}
~~~~~~
\begin{tikzpicture}
[line cap=round,line join=round,>=latex,x=0.7cm,y=0.7cm]
\draw[->,line width=1.5pt] (3*cos{70},3*sin{70}) arc (70:31:3);
\draw[->,line width=1.5pt] (3*cos{30},3*sin{30}) arc (30:-9:3);
\draw[->,line width=1.5pt] (3*cos{-10},3*sin{-10}) arc (-10:-49:3);
\draw[dashed, line width=1.5pt] (3*cos{-50},3*sin{-50}) 
arc (-50:-130:3);
\draw[->,line width=1.5pt] (3*cos{-130},3*sin{-130}) arc (-130:-171:3);
\draw[->,line width=1.5pt] (3*cos{-170},3*sin{-170}) arc (-170:-209:3);
\draw[->,line width=1.5pt] (3*cos{-210},3*sin{-210}) arc (-210:-249:3);

\draw[fill=blue] (3*cos{30},3*sin{30}) circle (2.5pt);
\draw[fill=blue] (3*cos{-10},3*sin{-10}) circle (2.5pt);
\draw[fill=blue] (3*cos{-50},3*sin{-50}) circle (2.5pt);
\draw[fill=blue] (3*cos{70},3*sin{70}) circle (2.5pt);
\draw[fill=blue] (3*cos{-130},3*sin{-130}) circle (2.5pt);
\draw[fill=blue] (3*cos{-170},3*sin{-170}) circle (2.5pt);
\draw[fill=blue] (3*cos{150},3*sin{150}) circle (2.5pt);
\draw[fill=blue] (3*cos{110},3*sin{110}) circle (2.5pt);

\draw (3*cos{70},3*sin{70}) node[above]{$P_1$};
\draw (3*cos{30}+0.2,3*sin{30}) node[above]{$P_2$};
\draw (3*cos{-10},3*sin{-10}) node[right]{$P_3$};
\draw (3*cos{-50}+0.1,3*sin{-50}-0.1) node[below]{$P_4$};
\draw (3*cos{-130}-0.2,3*sin{-130}) node[below]{$P_{n-3}$};
\draw (3*cos{-170}-0.1,3*sin{-170}) node[left]{$P_{n-2}$};
\draw (3*cos{150}-0.1,3*sin{150}+0.1) node[left]{$P_{n-1}$};
\draw (3*cos{110},3*sin{110}) node[above]{$P_{n}$};

\draw (0,-4.2) node[above]{The tournament $\mathcal S_n^-$};
\end{tikzpicture}

\noindent
In this section, we are going to prove that 
an inversion paradox occurs  for the Massey 
and the Colley methods with the tournaments 
$\mathcal{Z}_{n,k,\ell}$ if $k>\ell>0$.  
We remark that the removal of player $P_n$ in the tournament
$\cR_n^+$ gives the tournament $\cS_{n-1}^+$. This implies 
$(\mathcal{Z}_{n,k,\ell})'=k \cS_{n-1}^+ + \ell \cS_{n-1}^-$.

\subsection{The Massey method.}
\label{subsec-RSMassey}

Let us first make some remarks about the linear system associated to the Massey method applied to  the tournament $\mathcal{Z}_{n,k,\ell}$. The associated matrix $\overline{M}$ has the following form:

\begin{small}
$$\overline{M}=\begin{pmatrix}
2k+\ell&-(k+\ell)&0&0&  \cdots &0 &-k\\
-(k+\ell)&2(k+\ell)&-(k+\ell)&0&\cdots&0& 0\\
0&-(k+\ell)&2(k+\ell)&-(k+\ell) &0& \cdots &0\\
\vdots & \ddots & \ddots & \ddots & \ddots & \ddots & \vdots\\
0& \cdots & 0 &-(k+\ell) &2(k+\ell)&-(k+\ell)&0\\
0&\cdots&0&0&-(k+\ell) &2(k+\ell)&-(k+\ell)\\
1&1&1&1&\ldots&1&1\\
\end{pmatrix}.
$$
\end{small}

\noindent Furthermore, the vector $\ovl \p_n$ is equal to $(\ell,0,\cdots,0)^T$
and, in particular,
it does not depend on $k$.  Due to the structure of 
 the linear system 
$\ovl M\r=\ovl p_n$, the coordinates of
the ranking vector \mbox{$\r=(r_1,r_2,\ldots, r_n)^T$}
satisfy the  recurrence relation 
$$-(k+\ell)r_i +2(k+\ell) r_{i+1} -(k+\ell)r_{i+2}=0 \iff r_i-2r_{i+1}+r_{i+2}=0.$$
 This implies the 
existence of real numbers $\a$ and $\bb$ such that:
$$r_i=\a+\bb i, \textrm{ for } i=1, \ldots, n.$$

\begin{remark}\label{remarque-banale-Massey}
We have
$$r_{i+1}>r_{i}\Longleftrightarrow
\a+\bb(i+1)-(\a+\bb i) >0 \Longleftrightarrow
\bb>0.$$
Therefore, if $\bb>0$  then $(r_i)_{i\geq 1}$ is a strictly increasing sequence
and if $\bb<0$ it is  a strictly decreasing sequence.
If $\bb=0$, the sequence is constant. 
\end{remark}

Moreover, when we consider the first and the last row of the system, we get
$$(S_M)\quad\quad
\begin{cases}
(2k+\ell) r_1-(k+\ell) r_2 -kr_n=\ell\\
r_1+r_2+\cdots+r_n=0.
\end{cases}
$$
This system allows us to obtain $(\a,\bb)$ in terms of $k$ and $\ell$.\\

Now, when $P_n$ is deleted, the linear system associated to the Massey method applied to  the tournament is $(\mathcal{Z}_{n,k,\ell})'=k \mathcal{S}_{n-1}^++\ell\mathcal{S}_{n-1}^-$ and the associated matrix $\overline{M}'$ has the following form

\begin{small}
$$\overline{M}'=\begin{pmatrix}
k+\ell&-(k+\ell)&0&0&  \ldots &0 &0\\
-(k+\ell)&2(k+\ell)&-(k+\ell)&0&\ldots&0& 0\\
0&-(k+\ell)&2(k+\ell)&-(k+\ell) &0& \ldots &0\\
\vdots & \ddots & \ddots & \ddots & \ddots & \ddots & \vdots\\
0& \ldots & 0 &-(k+\ell) &2(k+\ell)&-(k+\ell)&0\\
0&\ldots&0&0&-(k+\ell) &2(k+\ell)&-(k+\ell)\\
1&1&1&1&\ldots&1&1\\
\end{pmatrix}.
$$
\end{small}

In this situation the cumulative point differential vector is $\p'=(\ell-k,0,\ldots,0)^T$.\\
As before, due to the structure of the linear system $\overline{M}'\r'=\p'$, the coordinates of the vector $\r'=(r'_1, \ldots,r'_{n-1})$ satisfy the recurrence relation 
$$r'_i-2r'_{i+1}+r'_{i+2}=0.$$
 Therefore, there exist real numbers $\alpha'$ and $\beta'$ such that 
 $$r'_i=\alpha'+\beta'i, \textrm{ for } i=1, \ldots, n-1$$ and we  have $r'_{i+1}>r'_i \iff \beta'>0$.\\
Moreover, the first and the last row of this system gives 
$$
(S'_M)\quad\quad
\begin{cases}
(k+\ell) r'_1-(k+\ell) r'_2 =\ell-k \\
r'_1+r'_2+\cdots+r'_{n-1}=0 .
\end{cases}
$$
Thanks to this system we can write $\alpha'$ and $\beta'$ in terms of $k$ and $\ell$.

\begin{proposition}\label{prop:Massey-ZRS}
For all positive integers $n$, $k$ and $\ell$, 
an inversion paradox occurs for the Massey method
applied to the tournament
$\mathcal{Z}_{n,k,\ell}$ if and only if $k>\ell>0$.
\end{proposition}

\begin{proof}
Let $\r=(\a + \bb i)_{1\leq i \leq n}$ be  the ranking vector
obtained with the Massey method applied 
to the tournament $\mathcal{Z}_{n,k,\ell}$.
The first row of $(S_M)$ gives
$\disp \bb=\frac{-\ell}{\ell+kn}<0$. Indeed, the first row gives
$$(2k+\ell) (\a+\bb)-(k+\ell) (\a+2\bb) -k(\a+n\bb)=\ell
\Longleftrightarrow
(-\ell-kn)\bb=\ell.$$
Therefore,  since $\ell>0$, we have $\beta <0$. By Remark \ref{remarque-banale-Massey},
 we  conclude that \linebreak $P_1 \succ P_2 \succ \cdots \succ P_n$.\\

When the player $P_n$ is removed,
we get the tournament 
\mbox{$(\mathcal{Z}_{n,k,\ell})'=k \cS_{n-1}^+ + \ell \cS_{n-1}^-$}.\\
In the same way as before, but with the system $S'_M$ instead of $S_M$, we get $\disp \bb'=\frac{k-\ell}{\ell+k}$.  Thus $\beta'>0$ if and only if $k>\ell$. This gives 
\mbox{$P_{n-1} \succ P_{n-2}  \succ \cdots \succ P_2 \succ P_1$}
if and only if $k>\ell$.\\
 Therefore, we conclude that an inversion paradox occurs if and only if $k>\ell>0$.
\end{proof}

\begin{example}
We will study the tournament $\mathcal{Z}_{n,2,1}=
2\mathcal R_n^+ + \mathcal{S}_n^-$ :
\end{example}

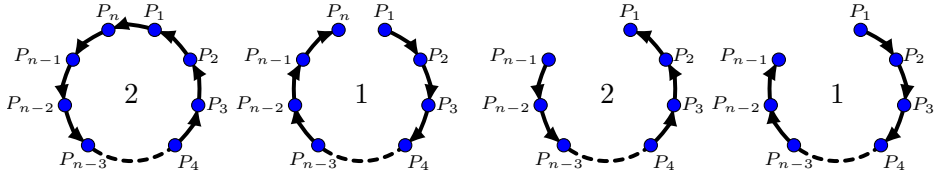
\begin{center}
\begin{tikzpicture}
[line cap=round,line join=round,>=latex,x=0.3cm,y=0.3cm]
\clip(-5.5,-5.8) rectangle (14.5,4.5);
\draw[->,line width=1.5pt] (3*cos{70},3*sin{70}) arc (70:109:3);
\draw[->,line width=1.5pt] (3*cos{30},3*sin{30}) arc (30:69:3);
\draw[->,line width=1.5pt] (3*cos{-10},3*sin{-10}) arc (-10:29:3);
\draw[->,line width=1.5pt] (3*cos{-50},3*sin{-50}) arc (-50:-11:3);
\draw[dashed, line width=1.5pt] (3*cos{-50},3*sin{-50}) 
arc (-50:-130:3);

\draw[->,line width=1.5pt] (3*cos{-170},3*sin{-170}) arc(-170:-131:3);
\draw[->,line width=1.5pt] (3*cos{-210},3*sin{-210}) arc(-210:-171:3);
\draw[->,line width=1.5pt] (3*cos{-250},3*sin{-250}) arc(-250:-211:3);

\begin{scriptsize}
\draw[fill=blue] (3*cos{30},3*sin{30}) circle (2.5pt);
\draw[fill=blue] (3*cos{-10},3*sin{-10}) circle (2.5pt);
\draw[fill=blue] (3*cos{-50},3*sin{-50}) circle (2.5pt);
\draw[fill=blue] (3*cos{70},3*sin{70}) circle (2.5pt);
\draw[fill=blue] (3*cos{-130},3*sin{-130}) circle (2.5pt);
\draw[fill=blue] (3*cos{-170},3*sin{-170}) circle (2.5pt);
\draw[fill=blue] (3*cos{150},3*sin{150}) circle (2.5pt);
\draw[fill=blue] (3*cos{110},3*sin{110}) circle (2.5pt);

\draw (3*cos{70},3*sin{70}+0.1) node[above]{$P_1$};
\draw (3*cos{30}+0.8,3*sin{30}-0.5) node[above]{$P_2$};
\draw (3*cos{-10},3*sin{-10}) node[right]{$P_3$};
\draw (3*cos{-50}+0.6,3*sin{-50}-0.1) node[below]{$P_4$};
\draw (3*cos{-130}-0.2,3*sin{-130}) node[below]{$P_{n-3}$};
\draw (3*cos{-170}-0.1,3*sin{-170}) node[left]{$P_{n-2}$};
\draw (3*cos{150}-0.1,3*sin{150}+0.1) node[left]{$P_{n-1}$};
\draw (3*cos{110},3*sin{110}+0.1) node[above]{$P_{n}$};
\end{scriptsize}

*****************   Deuxième roue

\foreach \a in {10.2}
{
\draw[->,line width=1.5pt] (3*cos{70}+\a,3*sin{70}) arc (70:31:3);
\draw[->,line width=1.5pt] (3*cos{30}+\a,3*sin{30}) arc (30:-9:3);
\draw[->,line width=1.5pt] (3*cos{-10}+\a,3*sin{-10}) arc (-10:-49:3);
\draw[dashed, line width=1.5pt] (3*cos{-50}+\a,3*sin{-50}) 
arc (-50:-130:3);
\draw[->,line width=1.5pt] (3*cos{-130}+\a,3*sin{-130}) arc (-130:-171:3);
\draw[->,line width=1.5pt] (3*cos{-170}+\a,3*sin{-170}) arc (-170:-209:3);
\draw[->,line width=1.5pt] (3*cos{-210}+ \a,3*sin{-210}) arc (-210:-249:3);

\begin{scriptsize}
\draw[fill=blue] (3*cos{30}+\a,3*sin{30}) circle (2.5pt);
\draw[fill=blue] (3*cos{-10}+\a,3*sin{-10}) circle (2.5pt);
\draw[fill=blue] (3*cos{-50}+\a,3*sin{-50}) circle (2.5pt);
\draw[fill=blue] (3*cos{70}+\a,3*sin{70}) circle (2.5pt);
\draw[fill=blue] (3*cos{-130}+\a,3*sin{-130}) circle (2.5pt);
\draw[fill=blue] (3*cos{-170}+\a,3*sin{-170}) circle (2.5pt);
\draw[fill=blue] (3*cos{150}+\a,3*sin{150}) circle (2.5pt);
\draw[fill=blue] (3*cos{110}+\a,3*sin{110}) circle (2.5pt);

\draw (3*cos{70}+\a,3*sin{70}+0.1) node[above]{$P_1$};
\draw (3*cos{30}+0.8+\a,3*sin{30}-0.5) node[above]{$P_2$};
\draw (3*cos{-10}+\a,3*sin{-10}) node[right]{$P_3$};
\draw (3*cos{-50}+0.6+\a,3*sin{-50}-0.1) node[below]{$P_4$};
\draw (3*cos{-130}-0.2+\a,3*sin{-130}) node[below]{$P_{n-3}$};
\draw (3*cos{-170}-0.1+\a,3*sin{-170}) node[left]{$P_{n-2}$};
\draw (3*cos{150}-0.1+\a,3*sin{150}+0.1) node[left]{$P_{n-1}$};
\draw (3*cos{110}+\a,3*sin{110}+0.1) node[above]{$P_{n}$};
\end{scriptsize}

\draw (\a,0) node{\large 1};
}
\draw (0,0) node{\large 2};
\draw (4,-5.8) node[above]{The tournament 
$\mathcal Z_{n,2,1}=2\mathcal R_n^+ + \mathcal S_n^-$};
\end{tikzpicture}
~
\begin{tikzpicture}
[line cap=round,line join=round,>=latex,x=0.3cm,y=0.3cm]
\clip(-6.2,-5.8) rectangle (14.5,4.5);
\draw[->,line width=1.5pt] (3*cos{30},3*sin{30}) arc (30:69:3);
\draw[->,line width=1.5pt] (3*cos{-10},3*sin{-10}) arc (-10:29:3);
\draw[->,line width=1.5pt] (3*cos{-50},3*sin{-50}) arc (-50:-11:3);
\draw[dashed, line width=1.5pt] (3*cos{-50},3*sin{-50}) 
arc (-50:-130:3);

\draw[->,line width=1.5pt] (3*cos{-170},3*sin{-170}) arc(-170:-131:3);
\draw[->,line width=1.5pt] (3*cos{-210},3*sin{-210}) arc(-210:-171:3);

\begin{scriptsize}
\draw[fill=blue] (3*cos{30},3*sin{30}) circle (2.5pt);
\draw[fill=blue] (3*cos{-10},3*sin{-10}) circle (2.5pt);
\draw[fill=blue] (3*cos{-50},3*sin{-50}) circle (2.5pt);
\draw[fill=blue] (3*cos{70},3*sin{70}) circle (2.5pt);
\draw[fill=blue] (3*cos{-130},3*sin{-130}) circle (2.5pt);
\draw[fill=blue] (3*cos{-170},3*sin{-170}) circle (2.5pt);
\draw[fill=blue] (3*cos{150},3*sin{150}) circle (2.5pt);

\draw (3*cos{70},3*sin{70}+0.1) node[above]{$P_1$};
\draw (3*cos{30}+0.8,3*sin{30}-0.5) node[above]{$P_2$};
\draw (3*cos{-10},3*sin{-10}) node[right]{$P_3$};
\draw (3*cos{-50}+0.6,3*sin{-50}-0.1) node[below]{$P_4$};
\draw (3*cos{-130}-0.2,3*sin{-130}) node[below]{$P_{n-3}$};
\draw (3*cos{-170}-0.1,3*sin{-170}) node[left]{$P_{n-2}$};
\draw (3*cos{150}-0.1,3*sin{150}+0.1) node[left]{$P_{n-1}$};
\end{scriptsize}

*****************   Deuxième roue

\foreach \a in {10.2}
{
\draw[->,line width=1.5pt] (3*cos{70}+\a,3*sin{70}) arc (70:31:3);
\draw[->,line width=1.5pt] (3*cos{30}+\a,3*sin{30}) arc (30:-9:3);
\draw[->,line width=1.5pt] (3*cos{-10}+\a,3*sin{-10}) arc (-10:-49:3);
\draw[dashed, line width=1.5pt] (3*cos{-50}+\a,3*sin{-50}) 
arc (-50:-130:3);
\draw[->,line width=1.5pt] (3*cos{-130}+\a,3*sin{-130}) arc (-130:-171:3);
\draw[->,line width=1.5pt] (3*cos{-170}+\a,3*sin{-170}) arc (-170:-209:3);

\begin{scriptsize}
\draw[fill=blue] (3*cos{30}+\a,3*sin{30}) circle (2.5pt);
\draw[fill=blue] (3*cos{-10}+\a,3*sin{-10}) circle (2.5pt);
\draw[fill=blue] (3*cos{-50}+\a,3*sin{-50}) circle (2.5pt);
\draw[fill=blue] (3*cos{70}+\a,3*sin{70}) circle (2.5pt);
\draw[fill=blue] (3*cos{-130}+\a,3*sin{-130}) circle (2.5pt);
\draw[fill=blue] (3*cos{-170}+\a,3*sin{-170}) circle (2.5pt);
\draw[fill=blue] (3*cos{150}+\a,3*sin{150}) circle (2.5pt);

\draw (3*cos{70}+\a,3*sin{70}+0.1) node[above]{$P_1$};
\draw (3*cos{30}+0.8+\a,3*sin{30}-0.5) node[above]{$P_2$};
\draw (3*cos{-10}+\a,3*sin{-10}) node[right]{$P_3$};
\draw (3*cos{-50}+0.6+\a,3*sin{-50}-0.1) node[below]{$P_4$};
\draw (3*cos{-130}-0.2+\a,3*sin{-130}) node[below]{$P_{n-3}$};
\draw (3*cos{-170}-0.1+\a,3*sin{-170}) node[left]{$P_{n-2}$};
\draw (3*cos{150}-0.1+\a,3*sin{150}+0.1) node[left]{$P_{n-1}$};
\end{scriptsize}

\draw (\a,0) node{\large 1};
}
\draw (0,0) node{\large 2};
\draw (4,-5.8) node[above]{The tournament 
$\mathcal Z'_{n,2,1}=2\mathcal S_{n-1}^+ + \mathcal S_{n-1}^-$};
\end{tikzpicture}
\end{center}

\noindent 
By keeping the notation introduced in Section \ref{subsec-RSMassey},
we get the ranking vector \mbox{$\r=(r_i)_{1\leq i \leq n}$}
where $r_i=\a+\bb i$
with $\disp \bb=\frac{-\ell}{\ell+kn}=\frac{-1}{1+2n}$.
With the second equation of $(S_M)$, we find that 
$\a=\dfrac{n+1}{2(2n+1)}$. Indeed, we have
\begin{eqnarray*}
r_1+r_2+\cdots+r_n=0 &\iff & (\a+\bb) +(\a+2\bb) + \cdots + (\a+n \bb)=0\\
&\iff & n\a + \Big( \dfrac{n(n+1)}{2}\Big) \bb =0\\
&\iff & n \a = - \dfrac{n(n+1)}{2} \bb = \dfrac{n(n+1)}{2(2n+1)}.
\end{eqnarray*} 

Now, we consider $(\mathcal{Z}_{n,2,1})'=  
2 \cS_{n-1}^+ + \cS_{n-1}^-$, the tournament
obtained after the deletion of the player $P_n$ 
in the tournament  $\mathcal{Z}_{n,2,1}$. 
%
In this situation, 
we get the ranking vector $\r'=(r'_i)_{1\leq i \leq n}$
where $r'_i=\a'+\bb' i$
with $\disp \bb'=\frac{k-\ell}{\ell+k}=\frac{1}{3}$.
Thanks to the second equation of $(S'_M)$, we find that 
$\a'=-\dfrac{n}{6}$ because

\begin{eqnarray*}
r'_1+r'_2+\cdots+r'_{n-1}=0 & \iff & 
(\a'+\bb')+\cdots + \big(\a'+(n-1))\bb'\big)=0 \\
&\iff & (n-1) \a' +\dfrac{n(n-1)}{2} \bb' =0\\
& \iff & \a'= -\dfrac{n}{2} \bb'=-\dfrac{n}{6}.
\end{eqnarray*}

\noindent In conclusion, as $\beta<0$, by Remark~\ref{remarque-banale-Massey}, we found a strictly decreasing 
ranking vector $\r=(r_1,\ldots,r_n)$ for the tournament $\mathcal{Z}_{n,2,1}$
with $\disp r_i=\frac{1}{2(1+2n)}\left(-2i+n+1\right)$ and
a strictly increasing  ranking vector $\r'=(r'_1,\ldots,r'_{n-1})$ 
for the tournament $(\mathcal{Z}_{n,2,1})'$
with $\disp r'_i=\frac{1}{6}\left(2i-n\right)$.
In other terms, an inversion phenomenom occured : the ranking
for $\mathcal{Z}_{n,2,1}$ is \mbox{$P_1 \succ P_2 \succ \cdots \succ P_n$}
while it is $P_{n-1} \succ P_{n-2} \succ \cdots \succ P_1$
for $(\mathcal{Z}_{n,2,1})'$.

\subsection{The Colley method.}
\label{subsec_RSColley}
Let us first make a few remarks about Colley's method
 applied to the $\mathcal{Z}_{n,k,\ell}$ tournament.
We have to solve a system $C\r=\b$
where $C$ is the  Colley matrix 
associated to $\mathcal{Z}_{n,k,\ell}$. Let us recall that $C=M+2I$ where $M$ is the matrix studied in the previous section on the Massey method and $\b=\mathbf{e} +\frac{1}{2}\p$
with $ \p=(\ell,0,\cdots,0,-\ell)^T$.
We remark that the vector $\b$ is independent of $k$. 
As  $C\mathbf{e}=2\mathbf{e}$, instead of solving $C\r=\b$, we solve $C\tilde{\r}=\p$, with \mbox{$\r=\frac{1}{2}(\tilde{\r}+\mathbf{e})$}.\\

As done in the previous section, thanks to the structure of the matrix $C$, the coordinates of
 the ranking vector \mbox{$\t \r=(\t r_1,\t r_2,\cdots \t r_n)^T$}
 satisfy the  recurrence relation
\begin{equation*}\label{recurrence-RS-Colley}
-(k+\ell)\t r_i+2(k+\ell+1) \t r_{i+1}-(k+\ell) \t r_{i+2}=0,
\quad
{\rm for}\:\: i=1, \ldots, n
\end{equation*}
\noindent which  implies the 
existence of real numbers $\a$ and $\bb$ such that:
$$\t r_i=\a x^i+ \bb y^i, \quad \textrm{ for } i=1, \ldots, n$$
where $x$ and $y$ are the roots
of 
\begin{equation}
\tag{3}
(k+\ell)\l^2 -2(1+k+\ell)\l + k+\ell=0. \label{EqCaracteristique-RS-Colley}
\end{equation}
We get: 
$\disp x=1+\frac{1}{k+\ell}-\frac{\sqrt{2(k+\ell)+1}}{k+\ell}$ and
$\disp y=1+\frac{1}{k+\ell}+\frac{\sqrt{2(k+\ell)+1}}{k+\ell}$. \\
Therefore $0<x<1 <y$.\\

\begin{remark}\label{remarque-banale-Colley}
We have $\r=(r_1,r_2,\cdots, r_n)$ with $r_i=\frac{1}{2}(1+\t r_i)$.
Of course, $r_{i}>r_{i+1} $ if and only if $
\t r_i > \t r_{i+1}$ for all integers $i\geq 1$. Moreover:
$$\t r_i>\t r_{i+1}\Longleftrightarrow
\a x^i+\bb y^i> \a x^{i+1}+ \bb y^{i+1} \Longleftrightarrow
\a x^i(1-x) > \bb y^i(y-1).$$
 As $0<x<1<y$, we have   $x^i(1-x)>0$ and $y^i(y-1)>0$. Thus, the inequality $r_i>r_{i+1}$ is satisfied  if $\a>0>\bb$
and the reverse inequality 
$r_i<r_{i+1}$ is satisfied if $\a < 0 < \bb$.
\end{remark}

At last, as done for the Massey method, we can obtain $\alpha$ and $\beta$ in terms of $k$ and $\ell$ thanks to the system 
$$
\left\{
\begin{array}{l}
(2k+\ell+2) \t r_1-(k+\ell) \t r_2 -k \t r_n=\ell/2\\ \\
-k \t r_1-(k+\ell)\t r_{n-1}+(2k+\ell+2) \t r_n=-\ell/2.
\end{array}
\right.
$$

\medbreak
With all these remarks, we can prove the following proposition. The proof is based solely on elementary linear algebra. 

\begin{proposition}\label{prop:colley-ZRS}
For all positive integers $n$, $k$, and $l$,
an inversion paradox occurs for the Colley method
applied to the tournament
$\mathcal{Z}_{n,k,\ell}$ if and only if  $k>\ell>0$.
\end{proposition}

\begin{example} 
The matrix $C$ and the vectors $\b$ and $\r$ associated to the tournament $\mathcal{Z}_{5,2,1}$ are
$$C=
\begin{pmatrix}
7&-3&0&0&-2\\
-3&8&-3&0&0\\
0&-3&8&-3&0\\
0&0&-3&8&-3\\
-2&0&0&-3&7
\end{pmatrix},
\quad 
\b=\begin{pmatrix}
1.5\\1\\1\\1\\0.5
\end{pmatrix},
\quad
\r=\begin{pmatrix}
71/126 \\ 66/126\\ 63/126\\ 60/126 \\ 55/126
\end{pmatrix}.
$$
We have $r_1>r_2>\cdots>r_5$.\\

The matrix $C'$ and the vectors $\b'$ and $\r'$ associated to $(\mathcal{Z}_{5,2,1})'$  are
$$C'=\begin{pmatrix}
5&-3&0&0\\
-3&8&-3&0\\
0&-3&8&-3\\
0&0&-3&5
\end{pmatrix}
,\quad
\b'=\begin{pmatrix}
0.5\\1\\1\\1.5
\end{pmatrix}
,\quad 
\r'=\begin{pmatrix}
35/92 \\43/92\\ 49/92\\ 57/92
\end{pmatrix}.
$$
We have $r'_4>r'_3>r'_2>r'_1$ and we find the result announced in Proposition~\ref{prop:colley-ZRS}.
\end{example}
\subsection{Some other remarks.}$\,$\\
\indent $\bullet$ We cannot get the same kind of results for the Borda method. Indeed, if we consider the tournament $\mathcal{Z}_{n,2,1}=2\mathcal{R}_n^+ + \mathcal{S}_n^-$, then $P_1$, $P_2$, \ldots, $P_{n-1}$ win 3 matches and $P_n$ wins 2 matches. Thus, we have $P_1=P_2=\cdots=P_{n-1}\succ P_n$ and we do not have a strict order on the players. Furthermore, if we delete $P_n$ then $P_1$ wins only 1 match, $P_2$, \ldots, $P_{n-2}$ win 3 matches and $P_{n-1}$ wins 2 matches. Thus, if $P_n$ is removed we get $P_2=P_3= \cdots =P_{n-2}\succ P_{n-1} \succ P_1$. We do not have an inversion paradox.\\

$\bullet$ The tournaments $\mathcal{Z}_{5,2,1}$ do not give an 
inversion paradox  when we use the Markov method. 
Indeed,  the matrix used for this method is 
$$G=\begin{pmatrix}
0&1/3&0&0&2/3\\
1&0&1/3&0&0\\
0&2/3&0&1/3&0\\
0&0&2/3&0&1/3\\
0&0&0&2/3&0
\end{pmatrix}.
$$

An eigenvector  associated to the eigenvalue $1$ of \mbox{$\alpha G+\frac{1-\alpha}{5}N_5$} is
$$\r=\begin{pmatrix}
\frac{11 \alpha^{4}-9 \alpha^{3}-9 \alpha^{2}+81 \alpha +81}{18 \alpha^{4}-6 \alpha^{3}-27 \alpha^{2}+54 \alpha +81}\\
\frac{4 \alpha^{4}-\alpha^{3}+9 \alpha^{2}+36 \alpha +27}{6 \alpha^{4}-2 \alpha^{3}-9 \alpha^{2}+18 \alpha +27}\\
\frac{5 \alpha^{4}+5 \alpha^{3}+6 \alpha^{2}+27 \alpha +27}{6 \alpha^{4}-2 \alpha^{3}-9 \alpha^{2}+18 \alpha +27}\\
\frac{8 \alpha^{4}+\alpha^{3}-3 \alpha^{2}+27 \alpha +27}{6 \alpha^{4}-2 \alpha^{3}-9 \alpha^{2}+18 \alpha +27}\\
1
\end{pmatrix}.
$$
For all $\alpha \in ]0,1[$, this gives the ranking $P_2\succ P_3 \succ P_4 \succ P_1 \succ P_5$.

If $P_5$ is deleted,  the  matrix used by the Markov method is
$$G'=\begin{pmatrix}
0& 1/3&0&0\\
1&0&1/3&0\\
0&2/3&0&1\\
0&0&2/3&0
\end{pmatrix}.
$$
  
An eigenvector  associated to the eigenvalue $1$ of \mbox{$\alpha G'+\frac{1-\alpha}{4}N_4$} is 
$$\r'=
\begin{pmatrix}
-\frac{\alpha^{2}+6 \alpha -9}{2 \alpha^{2}-3 \alpha +9}\\
-\frac{3 \left(2 \alpha^{2}-\alpha -3\right)}{2 \alpha^{2}-3 \alpha +9}\\
-\frac{3 \left(\alpha^{2}-2 \alpha -3\right)}{2 \alpha^{2}-3 \alpha +9}\\
1
\end{pmatrix}.
$$
For all $\alpha \in [0.75,1[$, this gives the ranking: $P_3 \succ P_4 \succ P_2 \succ P_1$.\\

We remark that if $\alpha \in [0.75,1[$ and $P_5$ is deleted the order is perturbed. 
However, we do not get the inverse order between 
$P_1$, $P_2$, $P_3$ and $P_4$.\\

$\bullet$ It may be tempting to think that adding a tournament of the type $\mathcal{R}_n$ has no effect on the outcome of a tournament with Massey or Colley method. However, the following example shows that it is not the case.

We set $\mathcal{Y}=\{(P_1,P_3,1,0),(P_3,P_2,1,0),(P_2,P_4,1,0)\}$.
The Massey method applied to this tournament gives the ranking: $P_1\succ P_3 \succ P_2 \succ P_4$.
When we add $2\mathcal{R}^+_4$ to this tournament, that is to say when we consider the tournament $\mathcal{Y}+2\mathcal{R}^+_4$, the Massey method gives: $P_1 \succ P_2  \succ P_3 \succ P_4$. Thus the ranking is \mbox{modified}.\\

The same  phenomenon appears with the Colley method. When we apply this method to $\mathcal{Y}$ we get $P_1 \succ P_3 \succ P_2 \succ P_4$ and when we consider $\mathcal{Y}+2\mathcal{R}^+_4$, we get $P_1 \succ P_2 \succ P_3 \succ P_4$.

\section{Conclusion.}
First of all, we  should keep in mind that the paradoxes we have presented appear because 
we have added cycles to a perfect tournament. We can then ask ourselves if it is reasonable to construct an order on the players from the result of a tournament with many cycles. Indeed, the cycles create instability in the result. For example, if we consider the cycle $P_1 \succ \cdots \succ P_{i-1} \succ P_i \succ P_{i+1} \succ \cdots \succ P_n \succ P_1$, then removing the player $P_i$ benefits $P_{i+1}$ because he is no longer beaten by $P_i$ and this disadvantages $P_{i-1}$ who beat $P_i$.
 The choice of player eliminated therefore has a decisive impact. 
However, even if cycles create instability in the final ranking, the aim of sports tournaments is to provide a ranking between players. That is their raison d'être.  We have to accept this instability, but we also have to be aware of it: the first can become the last.\\

We have shown that if a natural ranking method can be reduced by Condorcet tournaments and satisfies the long tournament property then the inversion paradox can appear. Furthermore, these properties are satisfied by some classical methods and we have constructed different examples of the inversion paradox. These examples are based on tournaments with a special structure.
For each match, the score is $1-0$ or $0-1$. So the difference in points between the two players is as small as possible. It is therefore possible to assume that the players are of roughly the same value, and so the possibility of reversing the order is not so surprising. However, in some sports a result of $1-0$ is very common. This is the case in soccer, where after scoring a goal a team can remain on defence for the rest of the match in order to retain a victory. The average difference in points between two teams depends on the sport considered, and the examples we have constructed are better suited to sports where the difference is small. However, this approach remains relevant for any competition in which we simply consider victory or defeat.

Actually, the Borda, Colley, and Markov methods use only the number of wins and losses. The score of a match is not taken into account. Thus,  there is no loss of generality by considering small scores.  In constrast, the situation is not the same with the Massey method. Consider a tournament $\mathcal{V}$ where each player $P_i$ plays one time against $P_j$ when $i \neq j$ and  $P_i$ wins  with the score $3-0$ if and only if $i<j$. The Massey matrix of the tournament $\mathcal{V}+2\mathcal{C}_{P_5 \succ P_4 \succ \cdots \succ P_1}$
 is the same as the one given for the tournament $\mathcal{X}=\mathcal{T}_{P_1 \succ P_2 \succ \cdots \succ P_5}+ 2 \mathcal{C}_{P_5 \succ P_4 \succ \cdots \succ P_1}$  in Section~\ref{sec:2} but now $p=(12,6,0-6,-12)$ and the ranking is $P_1 \succ P_2 \succ P_3 \succ P_4 \succ P_5$.
In the same way, if $P_5$ is deleted then the associated Massey matrix is the matrix $M'$ given in Section~\ref{sec:2} but $p'=(3,1,-1,-3)$ and we get the ranking  $P_1 \succ P_2 \succ \cdots \succ P_5$. Therefore, in this situation the inversion paradox does not appear with the Massey method while it appears with the Borda, Colley and Markov methods. 

The previous comment does not mean that the Massey method is more robust to the inversion paradox than other methods. Indeed, there are tournaments where the inversion paradox appears with the Massey method and does not appear with the Colley method. For example, consider the tournament $\mathcal{E}$ where $P_1$ wins against $P_2$ with the score $8-0$, $P_3$ wins against $P_1$ with the score $1-0$ and $P_3$ plays two matches against $P_1$ and the result of these two matches is $1-0$. Applying the Massey method to the tournament $\mathcal{E}$, we get $P_1 \succ P_3 \succ P_2$. Then, if $P_2$ is deleted we get $P_3 \succ P_1$.
There is therefore an inversion paradox, which is not the case with the Colley method.
 Indeed, with this tournament $\mathcal{E}$, the Colley method gives the ranking $P_3 \succ P_1 \succ P_2$ and if $P_2$ is deleted we get $P_3 \succ P_1$.\\


One shortcoming of the examples in Section 2, independent of the sport considered and the points differential between the teams, is that they all require a large number of matches between the players. Indeed,  we have adopted the structure used in the proof of Theorem~\ref{thm:inv}. This structure provides a simple and short proof but requires a lot of matches between players. Examples in Section 2 are constructed with this structure and this enabled us to highlight the inversion paradox for the Borda, Massey and Colley methods, but these may not be very realistic. This is why, in Section 3, we studied situations where the number of matches played by each player is low. We then considered tournaments, denoted $\mathcal{Z}_{n,k,\ell}$, in which $n$ players play $k+\ell$ matches.
For these tournaments, if $k>\ell>0$ the inversion paradox is again realised. So if $k=2$ and $\ell=1$, the players each play 3 matches and the inversion paradox occurs.\\

Finally, to avoid the inversion paradox, we need to consider methods that do not satisfy one of the three reasonable properties used in this article. Which property do we want to do without? As it seems difficult to want a method that is not natural, we leave the reader with the choice between doing without Condorcet's tournament reducibility or doing without the long tournament property.
Alternatively, we can use methods that verify all three properties and accept the risk of the inversion paradox appearing. So the question remains:
 What is the probability of the paradox occurring at sports tournaments?
Are there any sports where this paradox is more frequent? What is the probability of the occurence of the inversion paradox for each method? The answer to the last question would tell us whether one method is more robust than another to the inversion paradox. These questions can be studied from a practical point of view by looking at the results of different tournaments. We can also assume that the results follow a certain probabilistic model and try to estimate this probability for each method.   The probability of certain paradoxes has been already computed in the field of social choice, see e.g. \cite{Nurmi,Mossel}. In general, the probability of a paradox occurring is small. It remains to study what happens in the field of sport and whether the inversion paradox is frequent or not.


\section{Acknowledgments.}
The authors thank the referees for their comments, which helped to increase the quality of the paper.


 

%
%


\end{document}